\documentclass[12pt]{article}
\usepackage{geometry, booktabs, subcaption}
\usepackage{titlesec}
\usepackage[T1]{fontenc}
\usepackage{authblk}
\usepackage{amssymb,amsmath}
\usepackage{graphics}
\usepackage{natbib}
\RequirePackage[colorlinks,citecolor=blue,urlcolor=blue,linkcolor=black]{hyperref}
\usepackage{mathrsfs,amsthm,amsmath}
\usepackage{amssymb,multirow}
\usepackage{amsfonts}
\usepackage{stmaryrd}
\usepackage{dsfont}

\usepackage{ucs}
\usepackage[utf8x]{inputenc}
\usepackage{amsmath}
\usepackage{amsfonts}
\usepackage{amssymb}
\usepackage{amsthm}
\usepackage{fontenc}
\usepackage{graphicx}
\usepackage{cleveref}
\usepackage{bm}
\usepackage{hyperref}

\usepackage{bbm}
\usepackage{bm}
\RequirePackage[OT1]{fontenc}
\usepackage{tikz}
\usepackage[misc,geometry]{ifsym}

\newcommand{\tsp}{{\operatorname{T-SP}}}
\newcommand{\dd}{\mathrm{d}}
\newcommand{\indicator}{\mathrm{I}}

\renewcommand{\mid}{\,|\,}

\allowdisplaybreaks

\newcommand{\R}{\mathbb{R}}
\newcommand{\E}{\mathds{E}}

\def\simind{\stackrel{\mbox{\scriptsize{ind}}}{\sim}}
\def\simiid{\stackrel{\mbox{\scriptsize{iid}}}{\sim}}

\newtheorem{theorem}{Theorem}[section]
\newtheorem{lemma}[theorem]{Lemma}
\newtheorem{definition}[theorem]{Definition}

\mathchardef\mhyphen="2D
\def\Gammad{\mathrm{Gamma}}
\def\W{\mathbb{W}}

\usepackage{xparse}
\NewDocumentCommand{\evalat}{sO{\big}mm}{%
  \IfBooleanTF{#1}
   {\mleft. #3 \mright|_{#4}}
   {#3#2|_{#4}}%
}

\usepackage{fancyhdr}
\usepackage{placeins}

\title{Transform-scaled process priors for trait allocations in Bayesian nonparametrics}

\author[1]{Mario Beraha}
\author[2]{Stefano Favaro}
\affil[1]{\normalsize{Department of Economics and Statistics, University of Torino}} 
\affil[2]{\normalsize{Department of Economics and Statistics, University of Torino and \newline Collegio Carlo Alberto}}

\begin{document}

\maketitle

\begin{abstract}
Completely random measures (CRMs) provide a broad class of priors, arguably, the most popular, for Bayesian nonparametric (BNP) analysis of trait allocations. As a peculiar property, CRM priors lead to predictive distributions that share the following common structure: for fixed prior's parameters, a new data point exhibits a Poisson (random) number of ``new'' traits, i.e., not appearing in the sample, which depends on the sampling information only through the sample size. While the Poisson posterior distribution is appealing for analytical tractability and ease of interpretation, its independence from the sampling information is a critical drawback, as it makes the posterior distribution of ``new" traits completely determined by the estimation of the unknown prior's parameters. In this paper, we introduce the class of transform-scaled process (T-SP) priors as a tool to enrich the posterior distribution of ``new" traits arising from CRM priors, while maintaining the same analytical tractability and ease of interpretation. In particular, we present a framework for posterior analysis of trait allocations under T-SP priors, showing that Stable T-SP priors, i.e., T-SP priors built from Stable CRMs, lead to predictive distributions such that, for fixed prior’s parameters, a new data point displays a negative-Binomial (random) number of ``new" traits, which depends on the sampling information through the number of distinct traits and the sample size. Then, by relying on a hierarchical version of T-SP priors, we extend our analysis to the more general setting of trait allocations with multiple groups of data or subpopulations. The empirical effectiveness of our methods is demonstrated through numerical experiments and applications to real data.

\end{abstract}


\section{Introduction}\label{sec1}

Bayesian nonparametric (BNP) analysis of trait allocations deals with data belonging to more than one group, referred to as traits, and exhibiting nonnegative (integer) levels of association to each trait \citep{Tit(08),Zho(12),Zho(14),Bro(15),Hea(16),Jam(17),Bro(18),Cam(18),Hea(20)}. For example, single cell expression data contain multiple genes with their corresponding expression levels from ancestral populations, members of a social network have friends to which they send multiple messages, documents contain different topics with their corresponding words. Trait allocations generalize both species allocations or clustering, where traits are understood as species' labels and each data point belongs to a single label identifying a cluster, and feature allocations, where traits are understood as features and each data point exhibits only a binary membership to multiple features \citep{Bro(13),Bro(13a)}. In particular, traits may be viewed as the natural generalization of features through the inclusion of nonnegative levels of memberships. BNP analysis of trait allocations has proved to be relevant in diverse fields of practical interest, including single cell analysis and topic modeling \citep{Zho(14),Roy(15)}, matrix factorization \citep{Zho(12)}, visual object recognition \citep{Tit(08)}, image segmentation analysis \citep{Bro(15)}, and network analysis \citep{Aye(21)}. We refer to \citet{Bro(18)} for a comprehensive account on BNP analysis of trait allocations, including the development of de Finetti type (representation) theorems of practical interest in BNP analysis.

Completely random measures (CRMs) provide a broad class of priors for BNP analysis of both feature and trait allocations, the most popular being the Beta process for features \citep{Gri(05),Teh(09)} and the Gamma process for traits \citep{Tit(08),Zho(12)}. In such a context, \citet{Jam(17)} presented a framework for posterior analysis, showing that all CRM priors lead to a peculiar predictive distribution for a new data point $Z_{n+1}$, given $n\geq1$ observable data points $Z_{1:n}=(Z_{1},\ldots,Z_{n})$ and fixed prior's parameters: i) $Z_{n+1}$ displays a Poisson (random) number of ``new" features/traits, i.e. features/traits not in $Z_{1:n}$, which depends on $Z_{1:n}$ only through the sample size $n$; ii) $Z_{n+1}$ displays an ``old" feature/traits, i.e. features/traits in $Z_{1:n}$, with a probability that depends on $Z_{1:n}$ through the empirical distribution of its levels of association in $Z_{1:n}$, and $n$. Such a predictive structure is inherited by the Poisson process formulation of CRMs \citep{Kin(93)}, and it provides a limitation of CRM priors, in the sense of a lack of flexibility both with respect to the form of the posterior distribution of ``new" features/traits and the use of the information of $Z_{1:n}$ in such a distribution. While the Poisson posterior distribution is appealing for analytical tractability and ease of interpretation, its independence from $Z_{1:n}$ is a critical drawback of CRM priors, as it makes the posterior distribution of ``new" features/traits completely determined by the estimation of the unknown prior's parameters. 

\subsection{Our contributions}

In this paper, we present a BNP approach to trait allocations, which relies on a novel class of priors with a more flexible predictive structure than CRMs. This is inspired by the work of \citet{Cam(23)} on feature allocations, which shows how scaled process (SP) priors, first introduced in \citet{Jam(15)}, allow to enrich the predictive distribution of the Beta process. In particular, we introduce the class of transform-scaled process (T-SP) priors as transformations of SP priors to deal with the more general trait allocations, in such a way that SP priors correspond to the identity transform, and we develop a framework for their posterior analysis. As a special case, we consider transformations of SP priors built from the Stable CRM \citep{Kin(75),Jam(15)}, referred to as Stable T-SP (ST-SP) priors, showing that they provide a sensible trade-off between enriching the predictive structure of CRM priors, and maintaining its analytical tractability and ease of interpretation. Precisely, we show that ST-SP priors lead to a predictive distribution for $Z_{n+1}$, given $Z_{1:n}$ and fixed prior's parameter, such that: i) $Z_{n+1}$ displays a negative-Binomial (random) number of ``new" traits, which depends on $Z_{1:n}$ only through the number of distinct traits and the sample size $n$; ii) $Z_{n+1}$ displays an ``old" trait, i.e., appearing in $Z_{1:n}$, with a probability that depends on $Z_{1:n}$ through the empirical distribution of its levels of association in $Z_{1:n}$, and $n$. An extension of our analysis is presented for the more general setting of trait allocations with multiple groups of data or subpopulations \citep{Mas(18),Jam(21)}, introducing a hierarchical version of T-SP priors, and developing its posterior analysis.

This is the first work to provide a comprehensive framework for BNP analysis of feature and trait allocations, as well as their generalizations to multiple groups of data, by relying on a broad class of nonparametric priors that allow enriching the predictive distribution of CRM priors, while maintaining analytical tractability and ease of interpretation. The effectiveness of our methods is demonstrated through some empirical analyses on synthetic data and real data. First, in two simulated scenarios, we consider the classical problem of predicting the number of ``new'' traits in additional unobservable samples, comparing ST-SP priors against both CRM priors and the Stable-Beta SP prior of \cite{Cam(23)}, the latter only considering the presence or absence of the traits disregarding the corresponding levels of association. Moreover, we consider an application of ST-SP priors on a problem of text classification with more than six thousand documents and 14 thousand unique words, proposing a nonparametric ``naive Bayes'' classifier along the similar lines of \cite{Zho(16)}. In both cases, we show how T-SP priors result in a better predictive performance. In general, our analyses demonstrate: i) the practical usefulness of trait allocation models, that, compared to the simpler feature allocation models, may be better suited to capture the data generating process; ii) the critical role of T-SP priors in inducing a more flexible predictive structure than nonparametric priors available in the literature.

\subsection{Related works}

The lack of flexibility in the predictive structure of CRM priors was first discussed by \citet{Mas(22)} in the context of BNP inference for the unseen-feature problem, namely the estimation of the number of hitherto unseen features that would be observed if $m\geq1$ additional samples were collected. In particular, they showed that all CRM priors lead to a Poisson posterior distribution for the number of unseen features, with such a distribution depending on the observable sample only through the sample size. This motivated the work of \citet{Cam(23)}, where SP priors are applied to enrich the posterior distribution of the number of unseen features, improving posterior inferences with respect to both estimation and uncertainty quantification. A similar scenario occurs in BNP inference for the unseen-species problem under the Dirichlet process (DP) prior \citep{Fer(73)}, which led to the use of the Pitman-Yor process (PYP) prior \citep{Pit(95),Pit(97)} for enriching the posterior distribution of the number of unseen species, while maintaining the analytical tractability and ease of interpretation of the DP. See, e.g., \citet{Lij(07)} and \citet{Fav(09)} for details. The predictive structures of the DP and the PYP priors in species allocations somehow resemble that of CRM and ST-SP priors, respectively, in feature and trait allocations. To some extent, ST-SP priors may be viewed as the natural counterpart of the PYP prior in feature and trait allocations.

\subsection{Organization of the paper}

The paper is structured as follows. In Section \ref{sec2}, we introduce the class of T-SP priors and develop their posterior analysis, showing that the special case of ST-SP priors lead to a more flexible predictive structure than CRM priors, while maintaining analytical tractability and ease of interpretation. In section \ref{sec3}, we introduce the class of hierarchical T-SP priors for BNP analysis of trait allocations with multiple groups of data or subpopulations, and develop their posterior analysis. Section \ref{sec4} contains a numerical illustration of our BNP approach to trait allocations, whereas in Section \ref{sec5} we discuss our work and some directions for future research. Proofs and complementary results are deferred to the Supplementary Materials.

\section{T-SP priors for trait allocations}\label{sec2}

For a measurable space of traits $\mathbb{W}$, we assume $n\geq1$ observations to be modeled as a random sample $Z_{1:n}=(Z_{1},\ldots,Z_{n})$ from the stochastic process $Z(w)=\sum_{k\geq1}A_{k}\delta_{w_{k}}(w)$, as a process indexed by $w\in\mathbb{W}$, where $(w_{k})_{k\geq1}$ are traits in $\mathbb{W}$ and $(A_{k})_{k\geq1}$ are independent $\mathbb{N}_{0}$-valued random variables such that $A_{k}$ is distributed according to a distribution $G_{A}(\tau_{k})$, with $\tau_{k}\geq0$ being an unknown positive parameter, for $k\geq1$. The $\mathbb{N}_{0}$-valued observational process $Z$ is referred to as the trait process with score (distribution) $G_{A}$ and parameter (discrete measure) $\mu=\sum_{k\geq1}\tau_{k}\delta_{w_{k}}$, denoted as $\mathrm{TrP}(G_{A},\mu)$ \citep{Jam(17)}. BNP inference for trait allocations relies on the specification of a suitable prior distribution on $\mu$, leading to the model
\begin{align}\label{exch_mod}
Z_i\,|\,\mu & \quad\simiid\quad \mathrm{TrP}(G_{A},\mu) \qquad i=1,\ldots,n,\\[0.2cm]
\notag \mu& \quad\sim\quad\mathscr{M},
\end{align}
namely $\mu$ is a discrete random measure on $\mathds{W}$ whose law $\mathscr{M}$ takes on the interpretation of a nonparametric prior distribution for the unknown trait's composition of the population $Z$. By de Finetti's representation theorem, the random variables $Z_{i}$'s in \eqref{exch_mod} are exchangeable with directing (de Finetti) measure $\mathscr{M}$ \citep{Ald(85)}. See \cite{Cam(18)}, and references therein, for a detailed treatment on exchangeability in BNP trait allocation models.

Under the BNP model \eqref{exch_mod}, the marginal sampling process $(Z_{i})_{i\geq1}$ is referred to as the generalized Indian buffet process (IBP) \citep{Jam(17)}, with the IBP being the special case in which $Z$ is the $\{0,1\}$-valued Bernoulli process and $\mu$ is the Beta process \citep{Gri(11)}. As observed in \citet{Jam(17)}, for a given parameter $\mu$, a suitable model for the score $G_A$ can be represented as a spike-and-slab distribution, namely for $a\in\mathbb{N}_{0}$
\begin{displaymath}
    \text{Pr}[A_{i,k} = a; \tau_k] = G_A(a; \tau_k) = \pi_A(\tau_k) \tilde G_{A^\prime}(a;\tau_k) + (1 - \pi_A(\tau_k)) \delta_{\{0\}}(a),
\end{displaymath}
with $\tilde{G}_{A^\prime}$ being a distribution on $\mathbb{N}$, thus explicitly accounting for $A_{i,k} = 0$, i.e., observation $i$ not displaying the $k$-th trait. Equivalently $A_{i,k} = B_{i,k} A^\prime_{i,k}$, where $B_{i,k}$ is a Bernoulli random variable with parameter $\pi_A (\tau_k)$ and $A_{i,k}$ is distributed as $G_{A^\prime}(\tau_k)$, for $i,k\geq1$. The discreteness of $\mu$ entails that traits are shared by observations with positive probability. We denote by $k_n$ the number of distinct traits, labeled $w^*_1, \ldots, w^*_{k_n}$, each appearing with frequency $m_{l} = \sum_{1\leq i\leq n} I_{\mathbb{R}_{+}}(Z_i(w^*_l))$, for $l=1, \ldots, k_n$. Moreover, let $\mathcal{B}_l$ be the index set containing the observations displaying the $l$-th trait, i.e. $\mathcal{B}_l = \left\{ i \in \{1, \ldots, n\}\text{ : }Z_i(w^*_l) > 0  \right\}$, and let $\mathcal{A}_l = \left\{ A^\prime_{i, l}\text{ : } i \in \mathcal{B}_l \right\}$ be the set of values associated with the $l$-th trait, that is the integer values associated with the $l$-th trait, restricted to those observations who display it.

\subsection{A review of CRM priors}

CRM priors form a broad class of nonparametric priors for the parameter $\mu$ of the trait process $Z$ \citep{Jam(17),Bro(18)}. Consider a CRM $\mu$ on $\mathbb{W}$, which is a random element taking values on the space of bounded discrete measure on $(\mathbb{W},\mathcal{W})$, such that for $k\geq1$ and a collection of disjoint Borel sets $C_1, \ldots, C_k \in \mathcal{W}$ the random variables $\mu(C_{1}),\ldots,\mu(C_{k})$ are independent \citep{Kin(67)}. We consider CRMs of the form $\mu(\cdot) = \int_{\R_+} s  N(\dd s,\, \cdot) = \sum_{k \geq 1} \tau_k \delta_{W_k}(\cdot)$, where $ N = \sum_{k \geq 1} \delta_{(\tau_k, W_k)}$ is a Poisson random measure on $\R_+ \times \mathbb{W}$ with L\'evy intensity measure $\nu(\dd x, \dd w)$, which characterizes the distribution of $\mu$ in terms its random jumps $\tau_{k}$'s and random locations $W_{k}$'s \citep{Kin(67),Kin(93)}. We focus on homogeneous L\'evy intensity measures, namely measures of the form $\nu(\dd x, \, \dd w) = \theta \rho(x) \dd x \, B_0(\dd w)$ where $\theta > 0$ is a parameter, $B_0$ is a nonatomic probability measure on $\mathbb{W}$ and $\rho(x) \dd x$ is a measure on $\R_+$ such that $\int_{\R_+} \rho(\dd s) = +\infty$ and $\psi(u) := \int_{\R_+}(1 - e^{-us}) \rho(s) \dd s < +\infty$ for all $u>0$, which ensure that $0 < \mu (\mathbb{W}) < +\infty$ almost surely. We write $\mu \sim \mbox{CRM}(\theta, \rho, B_0)$. Under \eqref{exch_mod}, the law of $\mu$ provides provide a natural prior distribution for the parameter $\mu$ of the trait process $Z$. See \citet{Jam(17)} for a posterior analysis of CRM priors for trait allocations. In the next theorem, we recall the posterior distribution and the predictive distribution of CRM priors \citep[Theorem 3.1 and Proposition 3.2]{Jam(17)}.

\begin{theorem}\label{teo:post_crm}
Let $Z_{1:n}$ be a random sample under \eqref{exch_mod} with $\mu \sim \text{CRM}(\theta, \rho, B_0)$, such that $Z_{1:n}$ displays $k_{n}$ traits $\{W^*_1, \ldots, W^*_{k_{n}}\}$ with frequencies $(m_1, \ldots, m_{k_{n}})$  and associated index sets $\{\mathcal B_1, \ldots ,\mathcal B_{k_n}\}$. The posterior distribution of $\mu$ given $Z_{1:n}$ coincides with the distribution of
\[
\mu\,|\,Z_{1:n}\stackrel{\text{d}}{=}  \mu^\prime + \sum_{l=1}^{k_n} J^*_l \delta_{W^*_l},
\]
where $\mu^\prime$ is a CRM on $\mathbb{W}$ with L\'evy intensity measure $\theta \rho_n(s) \dd s B_0(\dd x)$ such that
$\rho_n(s) = \left(1 - \pi_A(s) \right)^n \rho(s)$, and the $J^{\ast}_{l}$'s are independent random jumps, also independent of $\mu^\prime $, with density function
\[
    f_{J^*_l}(s) \propto \left[ 1 - \pi_A(s) \right]^{n - m_l} \pi_A^{m_l}(s) \prod_{i \in \mathcal B_l} \tilde G_{A^\prime}(\dd a_{i, l}; s) \rho(s), \qquad s > 0.
\]
Furthermore, the predictive distribution of $Z_{n+1}$, given $Z_{1:n}$, coincides with the distribution of 
\begin{equation}\label{eq:pred_crm}
Z_{n+1}\,|\,Z_{1:n}\stackrel{\text{d}}{=}       Z^{\prime}_{n+1} + \sum_{l = 1}^{k_n} A_{n+1, l} \delta_{W^*_l},
\end{equation}
where $Z^{\prime}_{n+1}\,|\, \mu^\prime \stackrel{\text{d}}{=}\sum_{k\geq1}A^{\prime}_{n+1,k}\delta_{W^{\prime}_{k}}\sim\mathrm{TrP}(G_A,\mu^\prime)$ is independent of the  $A_{n+1, l}$'s, with the $A_{n+1, l}$'s being independent random variables such that $A_{n+1, l}\,|\, J^*_l \sim G_A(J^*_l)$ for any $l=1,\ldots,k_{n}$.
\end{theorem}

According to the predictive distribution \eqref{eq:pred_crm}, $Z_{n+1}$ displays ``new" traits $W_{k}^{\prime}$'s, i.e. traits not appearing in $Z_{1:n}$, and ``old" traits $W_{l}^{\ast}$'s, i.e. traits appeared in $Z_{1:n}$. Because of the assumption that $B_{0}$ is a non-atomic probability measure, the labels of the ``new" traits are different from the labels of the ``old" traits, with probability one. In particular, from \eqref{eq:pred_crm}, the conditional probability of observing ``new" traits, given $Z_{1:n}$, is determined by the law of $Z_{n+1} ^{\prime}$, i.e. 
\begin{displaymath}
   \Pr\left(Z_{n+1}(\mathbb W \setminus \{W^*_1, \ldots, W^*_{k_{n}}\}) > 0 \mid Z_{1:n} \right),
\end{displaymath}
which depends on $Z_{1:n}$ only through the sample size $n$, as it is clear from $\rho_n$. Still from \eqref{eq:pred_crm}, the conditional probability of observing an ``old" trait $W^{\ast}_{l}$, given $Z_{1:n}$, is determined by the law of $A_{n+1, l}$, i.e. 
\begin{displaymath}
    \Pr\left(Z_{n+1}(W^*_l) > 0 \mid Z_{1:n}\right) = \int \pi_A(s) f_{J^*_l}(s) \dd s, 
\end{displaymath}
which depends on  $Z_{1:n}$ through the sample size $n$, the empirical frequency $m_l$ of $W^{\ast}_{l}$, and the displayed scores $\mathcal A_l$. As a corollary of Theorem \ref{teo:post_crm}, the posterior distribution of the number of ``new" traits in $Z_{n+1}$, given $Z_{1:n}$ and fixed prior's parameters, is a Poisson distribution that depends on the information in $Z_{1:n}$ only through $n$. Such a posterior structure is peculiar to CRM priors, arising from the Poisson process formulation of CRMs \citep{Kin(93)}.

Although the class of CRM priors is broad, all CRM priors lead to the same Poisson predictive structure for the number of ``new" traits, which makes them not a flexible prior model with respect to the induced predictive distributions. While the Poisson distribution is appealing for making posterior inferences analytically tractable and of easy interpretability, its independence from $Z_{1:n}$ makes the BNP approach under CRM priors a questionable oversimplification, with the probability of generating ``new" traits being completely determined by the estimation of unknown prior’s parameters. Such a limitation of CRM priors has been first investigated in the work of \citet{Cam(23)}, where a solution has been proposed in the special case of feature allocations, by relying on the class of SP priors   \citep{Jam(15)}. In particular, \citet{Cam(23)} considered the Bernoulli process, i.e., $Z\sim\mathrm{TrP}(G_{A},\mu)$ with $G_{A}$ being the Bernoulli distribution, and showed that SP priors for $\mu$ lead to a richer predictive structure than CRM priors, possibly including the whole sampling information in terms of the number of distinct features and their corresponding frequencies. As an example, they focussed on a SP prior for which the probability of observing ``new" features depends on $Z_{1:n}$ only through the sample size $n$ and the number $k_{n}$ of observed features, whereas the probability of observing an ``old" feature $W^{\ast}_{l}$ depends on  $Z_{1:n}$ through the sample size $n$ and the empirical frequency $m_l$ of the feature $W^{\ast}_{l}$.

\subsection{T-SP priors: definition and posterior analysis}

T-SP priors generalize SP priors to deal with the general trait process $Z$, while maintaining their desirable predictive structure. For $\mu \sim \mbox{CRM}(\theta, \rho, B_0)$, we denote by $\Delta_1 > \Delta_2 > \cdots$ the decreasingly ordered random jumps $\tau_{k}$'s of $\mu$, and then define the discrete random measure
\begin{displaymath}
    \mu_{\Delta_1} = \sum_{k \geq 1} \frac{\Delta_{k+1}}{\Delta_1} \delta_{W_{k+1}},
\end{displaymath}
with $\Delta_{k+1}/\Delta_{1}\in(0,1)$ for $k\geq1$ and $\sum_{k\geq1}\Delta_{k+1}/\Delta_{1}<+\infty$. Moreover, let $f_{\Delta_1}(\zeta) = \exp\{\theta \int_\zeta^{+\infty} \rho(s) \dd s\} \theta \rho(\zeta)$ be the density function of the distribution of $\Delta_1$. Then for a  function $h: \R_+ \rightarrow \R_+$, we introduce the random variable $\Delta_{1, h}$ whose distribution has density
\begin{equation}\label{tilting}
    f_{\Delta_{1, h}}(\zeta) = h(\zeta) f_{\Delta_1}(\zeta),
\end{equation}
provided that $\int  f_{\Delta_{1, h}}(\zeta) \dd \zeta = 1$.
If $F_\zeta$ is the conditional distribution of $(\Delta_{k+1} / \Delta_{1})_{k \geq 1}$ given $\Delta_{1} = \zeta$, then a SP prior is defined as the law of $\mu_{\Delta_{1, h}} = \sum_{k \geq 1} \tau_k \delta_{W_{k+1}}$ where $(\tau_k)_{k \geq 1}$ is distributed as $F_{\Delta_{1,h}}$. Inherent to the construction of SP priors is that the random jumps $\tau_k$'s of $\mu_{\Delta_{1, h}}$ are $(0, 1)$-valued, which is a critical limitation with respect to the choice the score $G_A$ of the observational process. That is, differently from CRM priors, SP priors apply only to trait processes $Z\sim\mathrm{TrP}(G_{A},\mu)$ for which $G_{A}$ is parameterized through a $(0,1)$-valued parameter, e.g. the Bernoulli process. To overcome this limitation, we introduce the class of T-SP priors.

\begin{definition}\label{tsp}
Let $\mu \sim \mbox{CRM}(\theta, \rho, B_0)$,  $\Delta_1 > \Delta_2 > \cdots$ the decreasingly ordered random jumps $\tau_{k}$'s of $\mu$, and $\Delta_{1,h}$ be the random variable whose distribution has density function \eqref{tilting}. Moreover, for any diffeomorphism $T: (0, 1) \rightarrow (T_0, T_1) \subset \R_+$ let $\mathcal L_\zeta(\cdot)$ be the conditional distribution of $\left(T(\Delta_{k+1} / \Delta_1)\right)_{k \geq 1}$ given $\Delta_1 = \zeta$. A T-SP prior is the law of the discrete random measure
\begin{displaymath}
        \mu= \sum_{k \geq 1} \tau_k \delta_{W_k},
\end{displaymath}
where $(\tau_k)_{k \geq 1}$ is distributed as $\int_{\R_+} \mathcal L_{\zeta}(\cdot) f_{\Delta_{1, h}}(\zeta) \dd \zeta$, and $(W_k)_{k \geq 1}$ is independent of $(\tau_k)_{k \geq 1}$ with the $W_{k}$'s being independent and identically distributed according to $B_0$. We write $\mu \sim \mathrm{T \mhyphen SP}(\theta, \rho, h,T,B_{0})$.
\end{definition}

Definition \ref{tsp} generalizes the definition of SP prior, in the sense that SP priors are recovered by setting $T$ to be the identity function. As CRM priors, T-SP priors apply to arbitrary trait processes, that is processes $Z\sim\mathrm{TrP}(G_{A},\mu)$ for which $G_{A}$ is parameterized through a non-negative parameter. Now, we consider the BNP model \eqref{exch_mod} with $\mu \sim \mathrm{T \mhyphen SP} (\theta, \rho, h,T,B_{0})$ and present a posterior analysis of $\mu$, extending some of the main results of \citet{Jam(15)} and \citet{Cam(23)}. We start by providing sufficient conditions that ensure the finiteness of the number of traits in a random sample from \eqref{exch_mod}, which is a prerequisite in trait allocations. In particular, if $Z$ is a random sample under \eqref{exch_mod}, i.e.  $Z = \sum_{k \geq 1} A_k \delta_{W_k}$, then a sufficient condition for $\sum_{k \geq 1} I_{\mathbb{R}_{+}}(A_k)= \sum_{k \geq 1} B_k < +\infty$ is
\begin{equation}\label{eq:finite_traits}
    \E_{\Delta_{1, h}} \left[ \int_{T_0}^{T_1} \pi_a(s) \frac{\rho(\Delta_{1, h} T^{-1}(s))}{T^\prime(T^{-1}(s))} \Delta_{1, h} \dd s \right] < +\infty.
\end{equation}
See  \Cref{app:proof_finite_traits} for a proof. To present our main results on posterior analysis, it is useful to introduce:
\begin{itemize}
\item[i)]
\begin{equation}\label{eq:rho_k}
 \rho_k(s, \zeta) = \left(1 - \pi_A(s) \right)^{k} \frac{\zeta \rho(\zeta T^{-1}(s))}{T^\prime(T^{-1}(s))} I_{(T_0, T_1)}(s);
\end{equation}
\item[ii)]
\begin{equation}\label{eq:psi_k}
    \psi_k(\zeta) = \int_{T_0}^{T_1} \left[1 - (1 - \pi_A(s))^{k} \right]\frac{\zeta \rho(\zeta T^{-1}(s))}{T^\prime(T^{-1}(s))} \dd s.
\end{equation}
\end{itemize}
The next theorem characterizes, with respect to the latent random variable $\Delta_{1,h}$, the distribution of a random sample $Z_{1:n}$ from \eqref{exch_mod} with $\mu \sim\text{T-SP}(\theta, \rho, h, T, B_{0})$. Such a distribution is known as the exchangeable trait probability function \citep{Jam(17),Cam(18)}.

\begin{theorem}\label{marg_tsp}
Let $Z_{1:n}$ be a random sample under \eqref{exch_mod} with $\mu \sim \mathrm{T \mhyphen SP} (\theta, \rho, h,T,B_{0})$, such that $Z_{1:n}$ displays $k_{n}$ traits $\{W^*_1, \ldots, W^*_{k_{n}}\}$ with frequencies $(m_1, \ldots, m_{k_{n}})$ and associated index sets $\{\mathcal B_1, \ldots ,\mathcal B_{k_n}\}$. Then, the conditional distribution of $Z_{1:n}$ given $\Delta_{1,h}$ is of the form
    \begin{equation}\label{eq:marg}
          \theta^{k_{n}} e^{-\theta \psi_n(\Delta_{1, h})} \prod_{l=1}^{k_{n}} \int_{T_0}^{T_1} \rho_{n - m_l}(s, \Delta_{1, h}) \pi_A(s)^{m_l} \prod_{i \in \mathcal{B}_l} \left[\tilde G_{A^\prime}(\dd a_{i, k}; s) \right] \dd s B_0(\dd W^*_l).
\end{equation}
\end{theorem}

See \Cref{app:proof_marg} for the proof of Theorem \ref{marg_tsp}. The exchangeable trait probability function of a T-SP prior follows from \eqref{eq:marg} by integrating with respect to the distribution of $\Delta_{1,h}$, whose density function is \eqref{tilting}. The next theorem characterizes, still with respect to the latent variable $\Delta_{1,h}$, the posterior distribution and the predictive distribution of T-SP priors.

\begin{theorem}\label{teo:post_tsp}
Let $Z_{1:n}$ be a random sample under \eqref{exch_mod} with $\mu \sim \mathrm{T \mhyphen SP} (\theta, \rho, h,T,B_{0})$, such that $Z_{1:n}$ displays $k_{n}$ traits $\{W^*_1, \ldots, W^*_{k_{n}}\}$ with frequencies $(m_1, \ldots, m_{k_{n}})$ and associated index sets $\{\mathcal B_1, \ldots ,\mathcal B_{k_n}\}$. The conditional distribution of $\Delta_{1, h}$ given $Z_{1:n}$ has density function 
\begin{equation}\label{eq:post_delta}
   f_{\Delta_{1, h} \mid Z_{1:n}} (\zeta) \propto e^{-\theta \psi_n(\zeta)} \prod_{l=1}^{k_{n}} \int_{T_0}^{T_1} \rho_{n - m_l}(s, \zeta) \prod_{i \in \mathcal{B}_l} \left[\tilde G_{A^\prime}(\dd a_{i, l}; s) \dd s \right]  f_{\Delta_{1, h}}(\zeta).
\end{equation}
Moreover:
\begin{itemize}
\item[i)] the conditional distribution of $\mu$, given $Z_{1:n}$ and $\Delta_{1,h}$, coincides with the distribution of
\begin{equation}\label{post_tsp}
        \mu\,|\,(\Delta_{1, h},\,Z_{1:n})\stackrel{\text{d}}{=}\mu^\prime_{\Delta_{1, h}} + \sum_{l=1}^{k_{n}} J^*_l \delta_{W^*_l},
\end{equation}
    where $\mu^\prime_{\Delta_{1, h}}$ is a CRM on $\mathbb{W}$ with L\'evy intensity measure $\rho_{n}(s, \Delta_{1, h}) \dd s \theta B_0(\dd x)$, where $\rho_n$ is as in \eqref{eq:rho_k},
and the $J^*_l$'s are independent random jumps, also independent of $\mu^\prime_{\Delta_{1, h}}$, with density function
\begin{equation}\label{eq:post_jumps}
 f_{J^*_l}(s) \propto \pi_A(s)^{m_l} \rho_{n - m_l}(s, \Delta_{1, h}) \prod_{i \in \mathcal{B}_l} \tilde G_{A^\prime}(\dd a_{i, l}; s) I_{(T_0, T_1)}(s);
\end{equation}
\item[ii)] the conditional distribution of $Z_{n+1}$, given $Z_{1:n}$ and $\Delta_{1,h}$, coincides with the distribution of
\begin{equation}\label{pred_tsp}
        Z_{n+1}\,|\,(\Delta_{1, h},\,Z_{1:n})\stackrel{\text{d}}{=}Z^\prime_{n+1} + \sum_{l=1}^{k_{n}} A_{n+1, l} \delta_{W^*_l},
\end{equation}
    where $Z^\prime_{n+1}\,|\, \mu^\prime_{\Delta_{1, h}}=\sum_{k\geq1}A^{\prime}_{n+1,k}\delta_{W^{\prime}_{k}}\sim\mathrm{TrP}(G_A,  \mu^\prime_{\Delta_{1, h}})$ is independent of the $A_{n+1, l}$'s, with the $A_{n+1, l}$'s being independent random variables such that  $A_{n+1, l}\,|\,J^*_l \sim G_A( J^*_l)$, for any $l=1,\ldots,k_{n}$.
\end{itemize}
\end{theorem}

See \Cref{app:proof_post} for the proof of Theorem \ref{teo:post_tsp}. The posterior distribution and the predictive distribution of a T-SP prior follows directly from \eqref{post_tsp} and \eqref{pred_tsp}, respectively, by integrating with respect to the conditional distribution of $\Delta_{1, h}$ given $Z_{1:n}$, whose density function is displayed in \eqref{eq:post_delta}. In particular, this leads to a predictive distribution for which: i) the conditional probability of observing ``new" traits $W_{k}^{\prime}$'s, given  $Z_{1:n}$, is determined by the law of $(\Delta_{1,h},Z^{\prime}_{n+1})$; ii) the conditional probability of observing an ``old" trait $W_{l}^{\ast}$, given $Z_{1:n}$, is determined by the law of  $(\Delta_{1,h},A_{n+1, l})$. As the distribution of $(\Delta_{1,h},Z^{\prime}_{n+1})$ may include the whole information from $Z_{1:n}$, depending on the specification of $(\rho,h,T)$, the conditional probability of observing ``new"  traits, given $Z_{1:n}$, may also include such an information, i.e.
\begin{align*}
&\Pr\left(Z_{n+1}(\mathbb W \setminus \{W^*_1, \ldots, W^*_{k_{n}}\}) > 0\,|\, Z_{1:n} \right)\\
&\quad= \int \Pr(Z^{\prime}_{n+1}(\mathbb W) > 0 \mid \mu^\prime_{\Delta_1, h} = \eta) \mathcal{P}_{\mu^\prime_{\Delta_{1, h}} \mid \Delta_{1, h} = \zeta}(\dd \eta) f_{\Delta_{1, h} \mid Z_{1:n}}(\dd \zeta).     
\end{align*}
As a corollary of Theorem \ref{teo:post_tsp}, the posterior distribution of the number of ``new" traits in $Z_{n+1}$, given $Z_{1:n}$ and fixed prior's parameters, is a mixture of Poisson distributions that includes an amount of information from $Z_{1:n}$ that is completely determined by the mixing distribution, namely the conditional distribution of $\Delta_{1,h}$, given $Z_{1:n}$. Thus, T-SP priors enrich the Poisson posterior structure arising from CRM priors, leading to a more flexible posterior distribution that allows to include more sampling information than the sole sample size $n$.

\subsection{Examples}

We specialize our posterior analysis of T-SP priors to some popular choices of the score distribution $G_A$:  Bernoulli distribution, Poisson distribution, and negative-Binomial distribution. We focus on ST-SP priors, that is, we assume that $\mu$ is an $\alpha$-Stable CRM \citep{Kin(75)}. For $\alpha\in(0,1)$, this is a CRM with L\'evy intensity $\nu(\dd x,\, \dd x) = \alpha x^{-1 - \alpha} \theta B_0(\dd w)$, such that 
\begin{displaymath}
    f_{\Delta_1}(\zeta) = \theta \alpha \zeta^{-1 - \alpha} \exp(- \theta \zeta^{-\alpha}) \indicator_{\R_+}(\zeta)
\end{displaymath}
or equivalently, $\Delta_1^{- \alpha}$ is exponentially distributed with parameter $\theta$. See, e.g., \citet{Jam(15)}. Moreover, we assume a function $h$ that provides a polynomial tilting for $\Delta_{1, h}^{-\alpha}$, that is
\begin{displaymath}
    h_{c}(\zeta^{-\alpha}) = \frac{\theta^{c-1}}{\Gamma(c)} \zeta^{-\alpha(c - 1)},
\end{displaymath}
which leads to $\Delta_{1, h}^{-\alpha} \sim \mbox{Gamma}(c, \theta)$. Under these assumptions, we show how Theorem \ref{teo:post_tsp} leads to a simple posterior distribution of $\Delta_{1, h}$, given $Z_{1:n}$, and hence a simple posterior distribution and predictive distribution for ST-SP priors. The proofs of all the results presented below are deferred to \Cref{app:example_details}, where we also show an example where $G_A$ is a mixture of a continuous distribution with a point mass at zero, i.e., a spike and slab distribution.

\subsubsection{Bernoulli distribution}\label{sec:sbsp}

The case where $G_A$ is the Bernoulli distribution with parameter $s \in (0, 1)$ has been investigated in \cite{Cam(23)}. For completeness, we report the marginal, posterior, and predictive distributions for a ST-SP prior with $T$ being the identity transform.  Defining $\gamma_0^{(n)} = \alpha \sum_{i=1}^n B(1-\alpha, i)$, where $B(\cdot, \cdot)$ is the Beta function, we have $\Delta_{1, h}^{-\alpha}\,|\, Z_{1:n} \sim \Gammad(k + c + 1 ,\theta + \gamma_0^{(n)})$, so that the marginal distribution of $Z_{1:n}$ can be seen to be
\begin{displaymath}
    \frac{\alpha^k \theta^{c+1}}{(\theta + \gamma_0^{(N)})^{k_n + c +1}} \frac{\Gamma(k + c + 1)}{\Gamma(c+1)}\prod_{l=1}^{k_n} \frac{\Gamma(m_l - \alpha) \Gamma(n - m_l + 1)}{\Gamma(n - \alpha + 1)} B_0(\dd W^*_l).
\end{displaymath}
Furthermore, the jumps $J^{\ast}_{l}$ in \eqref{eq:post_jumps} are independent and Beta distributed with parameter $(m_l - \alpha, n - m_l + 1)$, and the random measure $\mu^\prime_{\Delta_{1, h}}$, conditionally to $\Delta_{1, h}$, is a CRM with L\'evy intensity
\begin{displaymath}
    \Delta_{1, h}^{-\alpha} (1-s)^n \alpha s^{-1-\alpha} \dd s \, \theta B_0(\dd x).
\end{displaymath}

\subsubsection{Poisson distribution}\label{sec:pois}

Now, we consider $G_{A}$ to be a Poisson distribution with parameter $rs$, for fixed $r > 0$ and $s \in \R_+$. In particular, in this case we have $\pi_A(s) = 1 - e^{-rs}$ and $\tilde G_{A^\prime}(a) = e^{-rs} (rs)^a \left(a! (1 - e^{-rs})\right)^{-1}$. It is easy to check that $T(s) := -\log(1-s)$ satisfies \eqref{eq:finite_traits}.
Moreover,
\begin{equation}\label{eq:i_rk}
    \psi_k(\zeta) = \alpha \zeta^{-\alpha} I(r, k), \qquad I(r, k) = \int_{\R_+} (1- e^{-rsk}) (1 - e^{-s})^{-1-\alpha} e^{-s} \dd s
\end{equation}
Then, the posterior distribution of $\Delta_{1,h}$ in \eqref{eq:post_delta} simplifies to $\Delta_{1, h}^{-\alpha}\,|\, Z_{1:n} \sim \mbox{Gamma}(c+k_n, \theta [ 1 + \alpha I(r, n)])$. By marginalizing with respect to $\Delta_{1, h}$ in \eqref{eq:marg}, we obtain the distribution of $Z_{1:n}$, i.e.,
\[
    \frac{\Gamma(c + k_n) \alpha^k \prod_{l=1}^{k_n} F(n, q_l, r, \alpha)}{\Gamma(c) \left(1 + \alpha I(r, n) \prod_{l=1}^{k_n} \prod_{i \in \mathcal B_l} a_{il}!\right)} \prod_{l=1}^{k_n} B_0(\dd W^*_l),
\]
where $q_l := \sum_{i \in \mathcal B_l} a_{i,l}$ and
\[
    F(n, c, r, \alpha) = \int_{\R_+} e^{- r s n} e^{-s} (1 - e^{-s})^{- 1 - \alpha} (rs)^c \dd s.
\]
It is also possible to marginalize with respect to $\Delta_{1, h}$ in the posterior distribution for the jumps $J^*_l$'s in \eqref{eq:post_jumps}, though the resulting density does not belong to a known parametric family.

From \cite{Jam(17)}, the number of ``new" traits $U^1_n$ displayed in $Z_{n+1}$ is distributed as a negative-Binomial distribution. In particular, $Z_{n+1}\,|\, \Delta_{1, h}$ displays a Poisson number of ``new" traits,  $U^1_n \mid \Delta_{1, h} \sim \mbox{Poi}(\theta \phi_{n+1}(\Delta_{1, h}))$, where $\phi_{n+1} = \psi_{n+1} - \psi_n$.
In our case, $\phi_k(\zeta) = \alpha \zeta^{-\alpha} \tilde I(r, k)$, where 
\begin{displaymath}
    \tilde I(r, k) = \int_{\R_+} (1 - e^{-rs}) e^{- rs(k+1)} e^{-s} (1 - e^{-s})^{- 1 - \alpha} \dd s.
\end{displaymath}
Since $\theta \phi_{n+1}(\Delta_{1,h})\,|\,Z_{1:n}$ is Gamma distributed with parameters $c+k_n$ and $\beta := (1 + \alpha I(r, n)) / (\alpha \tilde I(r, n+1))$, then  $U^1_n$ follows a negative-Binomial distribution with parameters $c + k_n$ and $p = \frac{\beta}{1 + \beta} = \frac{1 + \alpha I(r, n)}{1 + \alpha I(r, n+1)}$. The distribution of the number of ``new" traits $U^m_n$ displayed in $m\geq1$ additional samples $(Z_{n+1}, \ldots, Z_{n+m})$ can be derived in an analogous way. From \cite{Jam(17)}, given $Z_{1:n}$ and $\Delta_{1, h}$ each $Z_{n+j}$ displays $K^\prime_{n+j} \sim \mbox{Poi}(\theta \phi_{n+j}(\Delta_{1, h}))$ new features.
The $K^\prime_{n+j}$'s are conditionally independent given $\Delta_{1, h}$. Therefore, we can write
\[
    U^m_n = \sum_{j=1}^m K^{\prime}_{n+j}\,|\, \Delta_{1, h} \sim \mbox{Poi} \left(\theta \sum_{j=1}^m \phi_{n+j}(\Delta_{1, h}) \right) \equiv \mbox{Poi} \left(\theta \alpha \Delta_{1,h}^{-\alpha} (I(r, n+m) - I(r, n)) \right),
\]
Then, mixing the Poisson distribution with respect to $\Delta_{1, h}^{-\alpha} \sim \mbox{Gamma}(c+k_n, \theta [ 1 + \alpha I(r, n)])$ we obtain that $U^m_n$ follows a negative binomial distribution with parameters $c + k_n$ and $p =  \frac{1 + \alpha I(r, n)}{1 + \alpha I(r, n+m)}$.

\subsubsection{Negative-Binomial distribution}\label{sec:nb}

Following \cite{Jam(17)}, we consider $G_{A}$ to be a negative-Binomial distribution. That is, we assume
\begin{equation}\label{eq:nb_pmf}
    G_A(a) = \binom{a + r - 1}{a}(1 - e^{-s})^a e^{-sr}
\end{equation}
where $r >0$ is fixed and $s \in \R_+$. In particular, we have $\pi_A(s) = 1 - e ^{-sr}$ and $\tilde G_{A^\prime}(a) = G_A(a) / \pi_A(s)$. That is, $\pi_A$ coincides with the expression found for the Poisson example. Then, we proceed as in \Cref{sec:pois} and set $T(s) = -\log(1- s)$, which yields  the same expressions for $\rho_k(s, \zeta)$, $\psi_k(\zeta)$, $I(r, k)$.
Specifically, this entails that the posterior distribution for $\Delta_{1, h}$ is the same found in \Cref{sec:pois} and also the distribution for the number of new traits $U_{n}^m$ agrees with what was previously found. Instead, the marginal distribution of $Z_{1:n}$ equals
\begin{equation}\label{eq:marg_nb}
    \frac{\Gamma(c + k_n) \alpha^{k_n} \prod_{l=1}^{k_n} B(rn + 1, q_l - \alpha)}{\Gamma(c) \left(1 + \alpha \sum_{i=1}^n I(r, i)\right)^{c+k_n}} \prod_{l=1}^{k_n} \prod_{i \in B_l} \binom{a_{i, l} + r - 1}{a_{i, l}} B_0(\dd W^*_l),
\end{equation}
where $q_l := \sum_{i \in \mathcal B_l} a_{i,l}$.

\section{Hierarchical T-SP priors and multi-group IBP}\label{sec3}

We extend our analysis to the more general setting of trait allocations with multiple groups or subpopulations. In particular, we assume that observations are modeled as a random sample $(Z_{1,1},\ldots,Z_{1,n_{1}},\ldots,Z_{j,i},\ldots,Z_{J,1},\ldots,Z_{g,n_{g}})$, with $j$ being the index of population, such that
\begin{align}\label{eq:multi_ibp}
Z_{j,1},\ldots,Z_{j,n_{j}}\,|\,(\mu_1,\ldots \mu_g) & \quad\simiid\quad \mathrm{TrP}(G_{A},\mu_j)\\[0.2cm]
\notag (\mu_1, \ldots \mu_g) & \quad\sim\quad \mathscr{M}_g,
\end{align}
for $j=1,\ldots,g$, where the $\mu_{j}$'s are discrete random measures on $\W$. By means of de Finetti's representation theorem, the $Z_{j,i}$'s in \eqref{eq:multi_ibp} are partially exchangeable with directing (de Finetti) measure $\mathscr{M}_g$. In this context, \cite{Mas(18)} proposed a hierarchical CRM (hCRM) prior for $(\mu_1, \ldots \mu_g)$,  by first letting $\mu_0 =  \sum_{k \geq 1} J_{0, k} \delta_{W_k} \sim \mbox{CRM}(\theta, \rho_0, B_0)$ and then setting
\begin{displaymath}
    \mu_j \mid \mu_0 = \sum_{k \geq 1} J_{j, k} \delta_{W_k}, \quad J_{j,k} \mid \mu_0 \simiid f_{J}(\cdot; J_{0,k}, r_j), \quad j=1, \ldots, g,
\end{displaymath}
where $f_J$ is a density function on $\R_+$ with parameters $(J_{0,k}, r_j)$. For short, we write $(\mu_1, \ldots, \mu_g) \sim \mathrm{hCRM}(f_{J}, \theta \rho_0, B_0)$. We refer to \cite{Mas(18)} for a  BNP  analysis of hCRM priors in the context of trait allocations. In the next theorem, we recall the posterior distribution of hCRMs. Similarly to the non-hierarchical setting, we define the index set of the $l$-th trait as $\mathcal{B}_{j, l} = \left\{i \in \{1, \ldots, n_j\}  \mbox{ s.t. } Z_{j,i}(W^*_l) > 0  \right\}$ and denote by $m_{j,l}$ its cardinality.

\begin{theorem}\label{teo:post_hcrm}
    Let $(Z_{j, i},  j=1, \ldots, g, \, i = 1, \ldots, n_j)$ be a random sample under \eqref{eq:multi_ibp}, with $(\mu_1, \ldots \mu_g) \sim \mathrm{hCRM}(f_{J}, \theta \rho_0, B_0)$, such that the sample displays $k_n$ traits $\{W^*_1, \ldots, W^*_{k_n}\}$ with frequencies $(m_{j,l}, j=1, \ldots, g, \ l=1, \ldots, k_n)$ and associated index sets $\{\mathcal{B}_{1, 1},\ldots,\mathcal{B}_{g, k_{n}}\}$. Then, the conditional distribution of $(\mu_1, \ldots, \mu_g)$, given $\{Z_{j,i}\}_{j,i}$, coincides with the distribution of
\begin{displaymath}
\left( \mu^\prime_1 + \sum_{l=1}^{k_n} J^*_{1, l} \delta_{W^*_l}, \ldots,  \mu^\prime_g + \sum_{l=1}^{k_n} J^*_{g, l} \delta_{W^*_l}\right),
\end{displaymath}
where 
\begin{itemize}
\item[i)] $(\mu^\prime_1, \ldots, \mu^\prime_g) \sim \mathrm{hCRM}(f^\prime_{J}, \rho^\prime_0, B_0)$ with
\begin{displaymath}
       \rho^\prime_0(s_0) = \prod_{j=1}^g \int_{\R_+} (1 - \pi_A(s_j))^{n_j} f_{J}(s_j; s_0, r_j) \dd s_j \, \rho_0(s_0)
\end{displaymath}
and
\begin{displaymath}
    f^\prime_{J}(s; s_0, r_j) \propto (1 - \pi_A(s))^{n_j} f_{J}(s; s_0, r_j);
\end{displaymath}
\item[ii)] the jumps $\{J^*_{j, l}\}$ are such that
\begin{displaymath}
         J^*_{j, l} \mid J^*_{0, l} \simind f_{J^*_{j, l}}(s)  \propto   (1 - \pi_A(s))^{n_j - m_{j, l}} \pi_A(s)^{m_{j, l}} \prod_{i \in \mathcal{B}_{j, l}}  \tilde G_{A^\prime}(\dd a_{j, i, l} \mid s) f_{J}(s; J^*_{0, l}, r_j) 
\end{displaymath}
and
\begin{displaymath}
         J^*_{0, l} \simind f_{J^*_{0, l}}(s_0)  \propto \prod_{j=1}^g\left[ \int_{\R_+} g_{J^*_{j, l}}(s_j; s_0) \dd s_j \right] \rho_0(s_0).
\end{displaymath}
\end{itemize}   
\end{theorem}

Although not explicitly stated in \cite{Mas(18)}, from \Cref{teo:post_hcrm} it follows that the conditional distribution of a new observation $Z_{j, n_j + 1}$, given $\{Z_{j, i}\}_{j, i}$, coincides with the distribution of
\begin{equation}\label{eq:pred_hcrm}
    Z^\prime_{j, n_j+1} + \sum_{l=1}^k A_{j, n_j+1, l} \delta_{W^*_l},
\end{equation}
where $Z^\prime_{j, n_j+1} \mid \mu_j \sim \mathrm{TrP}(G_A; \mu^\prime_j)$ and the $A_{j, n_j+1, l} \mid J^*_{j, l} \simind G_A(\cdot; J^*_{j, l})$. According to \eqref{eq:pred_hcrm}, since $B_0$ is a non-atomic distribution,  the probability of observing new traits is completely determined by the distribution of $Z^\prime_{j, n_j+1}$. Such a distribution can be seen to depend on the observed sample only through the cardinalities $(n_1, \ldots, n_g)$ regardless of the choice of $\rho_0$ and $f_J$.

In a recent work, \cite{Jam(21)} proposed an alternative class of nonparametric priors for $(\mu_1, \ldots \mu_g)$, still based on a hierarchical construction from CRMs. They assume that
    \begin{align}\label{eq:hcrm_james}
        \mu_1, \ldots, \mu_g \mid \mu_0 &\quad \simiid\quad \mathrm{CRM}(\theta_{j},\rho_j, \mu_0)\\[0.2cm]
        \notag \mu_0 & \quad \sim \quad \mathrm{CRM}(\theta,\rho_0, B_0)
    \end{align}
for $j=1, \ldots, g$. This entails that $\mu_j = \sum_{k \geq 1} J_{j,k} \delta_{\omega_{j, k}}$ and $Z_{j, i} = \sum_{k \geq 1} A_{j,i,k} \delta_{\omega_{j, k}}$ where $\omega_{j, k} \mid \mu_0$ are i.i.d. from $\mu_0 / \mu_0(\mathbb{W})$ and $ A_{j,i,k} \mid \mu_j \simind G_A(J_{j,k})$, $j=1, \ldots, g$. Furthermore, by expanding $\mu_{0}$ as $\mu_0 = \sum_{k \geq 1} J_{0,k} \delta_{W_k}$, and then grouping all atoms that are equal, we can write
\begin{equation}\label{eq:hcrm1}
    \mu_j \stackrel{\text{d}}{=} \sum_{k \geq 1} \left[\sum_{l \geq 1} \tilde J_{j,k, l}\right] \delta_{W_k}, \quad Z_{j,i} \stackrel{\text{d}}{=}  \sum_{k \geq 1} \left[\sum_{l \geq 1} \tilde A_{j,i,k, l}\right] \delta_{W_k},
\end{equation}
where, by the properties of CRMs, $(\tilde J_{j, k, l})_{l \geq 1}$ are, conditionally to $\mu_0$, the jumps of a CRM with L\'evy intensity $J_{0,k} \rho_j(s) \dd s$ and $\tilde A_{j,i,k, l} \mid \tilde J_{j, k, l} \simiid G_A(\cdot; \tilde J_{j, k, l})$. Posterior inferences under \eqref{eq:hcrm_james} are not trivial, since evaluating the distribution of the $Z_{j,i}$'s requires the probability mass function of an infinite convolution of random variables distributed as $G_A$ with different parameters, which is in general non available. Both the marginal distribution and posterior distribution involve complex combinatorial objects whose evaluation is rather complex. Similarly to the case of the hierarchical DP prior \citep{Teh(06)}, posterior inferences can be carried out by introducing auxiliary parameters that help with Markov chain Monte Carlo sampling.

\subsection{Hierarchical T-SP priors}

To avoid the combinatorial hurdles arising from the hierarchical formulation of \cite{Jam(21)}, i.e. from \eqref{eq:hcrm1}, here we consider the hierarchical formulation of \cite{Mas(18)} to define the class of hierarchical T-SP (hT-SP) priors. The construction follows closely that of T-SPs.
\begin{definition}\label{def:htsp}
    Let $\mu_0 = \sum_{k \geq 1} \tilde \tau_{0,k} \delta_{W_k} \sim \mathrm{T \mhyphen SP} (\theta, \rho_0, h, T, B_0)$. Then a $\mathrm{hT \mhyphen SP}$  is the vector of discrete random measures $(\mu_1, \ldots, \mu_g)$ such that, conditionally to $\mu_0$, the $\mu_j$'s are independent and 
    \[
       \mu_j   = \sum_{k \geq 1} \tau_{j,k} \delta_{W_k}
    \]
    where $\tau_{j,k} \mid \mu_0 \simind f_{\tau}(\cdot; \tau_{0, k}, r_j)$, for some density function $f_{\tau}$ on $\R_+$. We write $(\mu_1, \ldots, \mu_g) \sim \mathrm{hT \mhyphen SP}(f_\tau, \theta, \rho_0, h, T, B_0)$.
\end{definition}

Consider $\mu_{0}$ introduced in Definition \ref{def:htsp}. From the definition T-SP priors in Section \ref{sec2}, we have that, conditionally on $\Delta_{1, h}$, $\mu_0$ is a CRM with L\'evy intensity $\theta \rho_0(s_0 \mid \Delta_{1, h}) \dd s_0 B_0(\dd x)$, where 
\begin{equation}\label{eq:cond_levy}
    \rho_0(s_0 \mid \Delta_{1, h}) =  \frac{\rho_0(\Delta_{1, h} T^{-1}(s_0))}{T^\prime(T^{-1})(s_0)} \Delta_{1, h} \indicator_{(T_0, T_1)}(s_0).
\end{equation}
See Lemma \ref{lemma:conditional_levy} in  \Cref{proofs} for the proof of \eqref{eq:cond_levy}. Accordingly, $(\mu_1, \ldots, \mu_g) \mid \Delta_{1, h} \sim \mathrm{hCRM}(f_\tau, \theta \rho_0(s_0 \mid \Delta_{1, h}), B_0)$.
For ease of notation, define $\eta_k(s; s_0, r) = (1 - \pi_A(s))^k f_\tau(s \mid s_0, r)$ and $\gamma_k(s_0, r) = \int_{\R_+}\eta_k(s; s_0, r) \dd s$. The next theorem, which follows from \cite{Mas(18)}, characterizes the posterior and marginal distribution of $(\mu_1, \ldots, \mu_g) \sim  \mathrm{hT \mhyphen SP}(f_\tau, \theta, \rho_0, h, T, B_0)$.

\begin{theorem}\label{teo:post_htsp}
    Let $(Z_{j, i},  j=1, \ldots, g, \, i = 1, \ldots, n_j)$ be a random sample under \eqref{eq:multi_ibp}, with $(\mu_1, \ldots, \mu_g) \sim  \mathrm{hT \mhyphen SP}(f_\tau, \theta, \rho_0, h, T, B_0)$, such that the sample displays $k_n$ traits $\{W^*_1, \ldots, W^*_{k_n}\}$ with frequencies $(m_{j,l}, j=1, \ldots, g, \ l=1, \ldots, k_n)$ and associated index sets $\{\mathcal{B}_{1, 1},\ldots,\mathcal{B}_{g, k_{n}}\}$. Then, the conditional distribution of the $Z_{j, i}$'s given $\Delta_{1,h}$ is
        \begin{align*}
            \Pi(\bm m, \bm a \mid \Delta_{1, h}) &= \theta^{k_n} \exp\left\{-\theta \int_{\R_+}\left(1 - \prod_{j=1}^g \gamma_{n_j}(s_0, r_j)\right) \rho_0(s \mid \Delta_{1, h}) \dd s_0\right\}\\
            &\quad\times\prod_{l=1}^k \left[\int_{T_0}^{T_1} \prod_{j=1}^g \int \pi_A(s_j)^{m_{j, l}} \eta_{n_j - m_{j, l}}(s_j; s_0, r_j)\right.\\
            &\quad\quad\quad\quad\times\left. \prod_{i \in \mathcal{B}_{j, l}} \tilde G_{A^\prime}(\dd a_{j, i, l}; s_j) \dd s_j \rho_0 (s_0\mid  \Delta_{1, h}) \dd s_0\right] B_0(\dd W^*_l).
         \end{align*}
    The conditional distribution of $(\mu_1, \ldots, \mu_g)$ and $\mu_0$ given $(Z_{j,i})_{j, i}$ and $\Delta_{1,h}$ is as in \Cref{teo:post_hcrm} where $\rho_0(s_0)$ is replaced by $\rho_0 (s_0\mid  \Delta_{1, h})$ as in \eqref{eq:cond_levy}. Moreover, the conditional distribution of $\Delta_{1, h}$ given $(Z_{j,i})_{j, i}$ has density function
\begin{displaymath}
        f_{\Delta_{1, h}} (\zeta) \propto \Pi(\bm m, \bm a \mid \zeta) f_{\Delta_{1, h}}(\zeta).
\end{displaymath}
\end{theorem}

The predictive distribution of $Z_{j, n_{j+1}}$, given the sample, follows by first conditioning on $\Delta_{1, h}$, therefore obtaining \eqref{eq:pred_hcrm}, and then by marginalizing with respect to the posterior distribution of $\Delta_{1, h}$. See Section \ref{sec2} for details. Such a construction leads to a more flexible predictive structure than the hCRM priors of \cite{Mas(18)}, since the posterior distribution of the number of ``new'' traits may depend on the whole sampling information, in analogy to the predictive structure of T-SP prior. We defer to future research the problem of investigating such a novel predictive structure. In particular, we refer to the problem of finding classes of $\rho_0$ and $f_\tau$ for which the posterior distribution of the number of ``new'' traits does not include the whole sampling information, but only the $n_j$'s and the number of distinct traits in each group, similar to what we presented in Section \ref{sec2} for T-SP priors.

\section{Numerical illustrations}\label{sec4}

We focus here on the negative-Binomial trait process with ST-SP prior discussed in Section \ref{sec:nb}, henceforth called NB-ST-SP.
In Section \ref{sec:simulation} we focus on the prediction of newly displayed traits, and compare the NB-ST-SP with a trait process with negative-Binomial score distribution and Gamma process prior (NB-Ga), i.e., $\mu \sim \mbox{CRM}(\theta, s^{-1} e^{-s}, B_0)$, and with the stable-Beta SP (SB-SP) in \cite{Cam(23)} also discussed in Section \ref{sec:sbsp}.
In Section \ref{sec:text} we propose a ``naive-Bayes'' nonparametric model for text classification and compare the use of NB-ST-SP with NB-Ga. In both cases, we adopt an empirical Bayesian approach and estimate the hyperparameters in the prior distributions by maximizing the (log) marginal likelihood of the data. 
To optimize with respect to the parameters, we use the BFGS algorithm in the \texttt{Python} package \texttt{jax}.

\subsection{Simulated data}\label{sec:simulation}

\begin{figure}[h!]
	\centering
    \begin{subfigure}[b]{0.5\textwidth}
         \centering
         	\includegraphics[width=\textwidth]{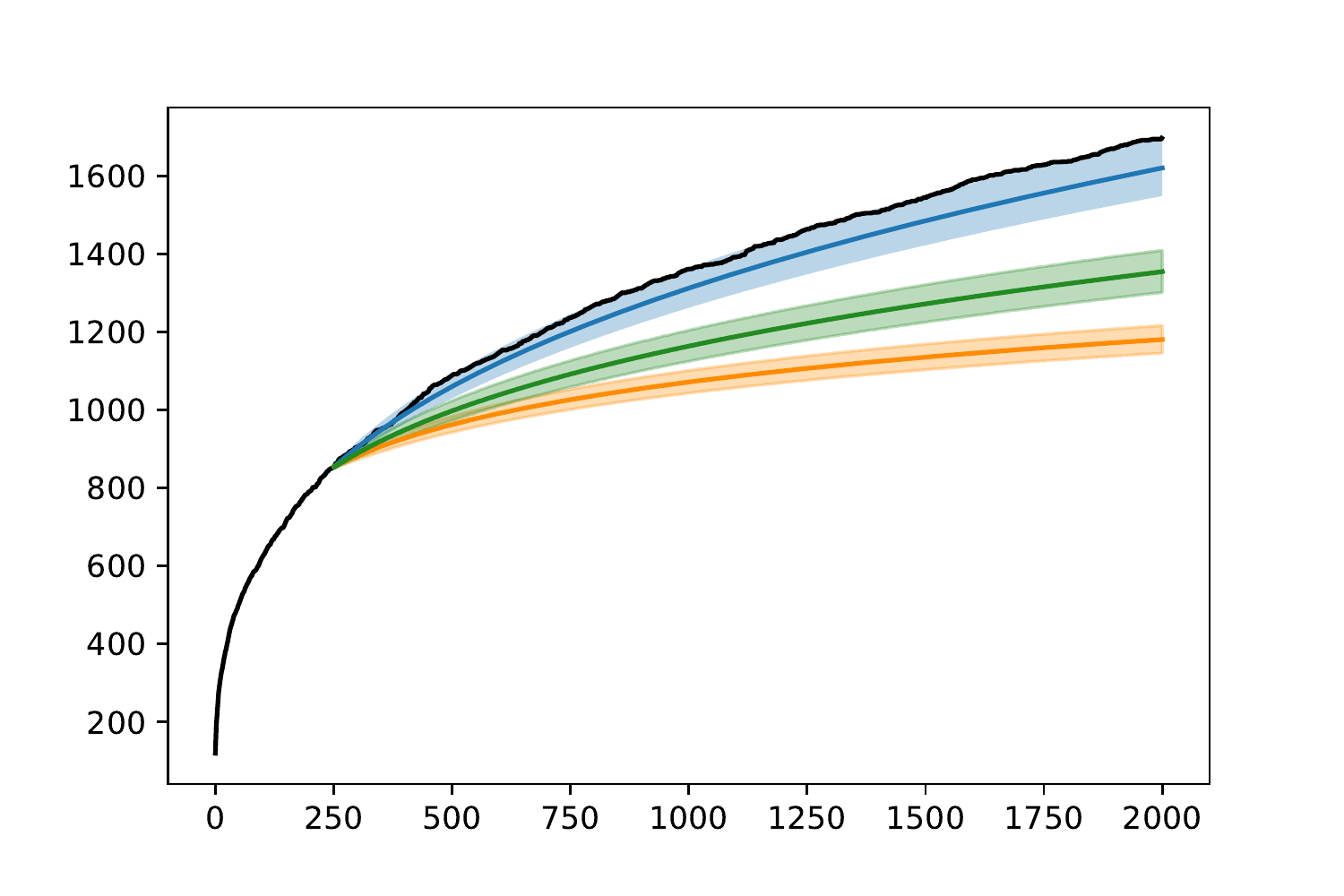}
	\caption{Simulation I}
     \end{subfigure}%
     \begin{subfigure}[b]{0.5\textwidth}
         \centering
         	\includegraphics[width=\textwidth]{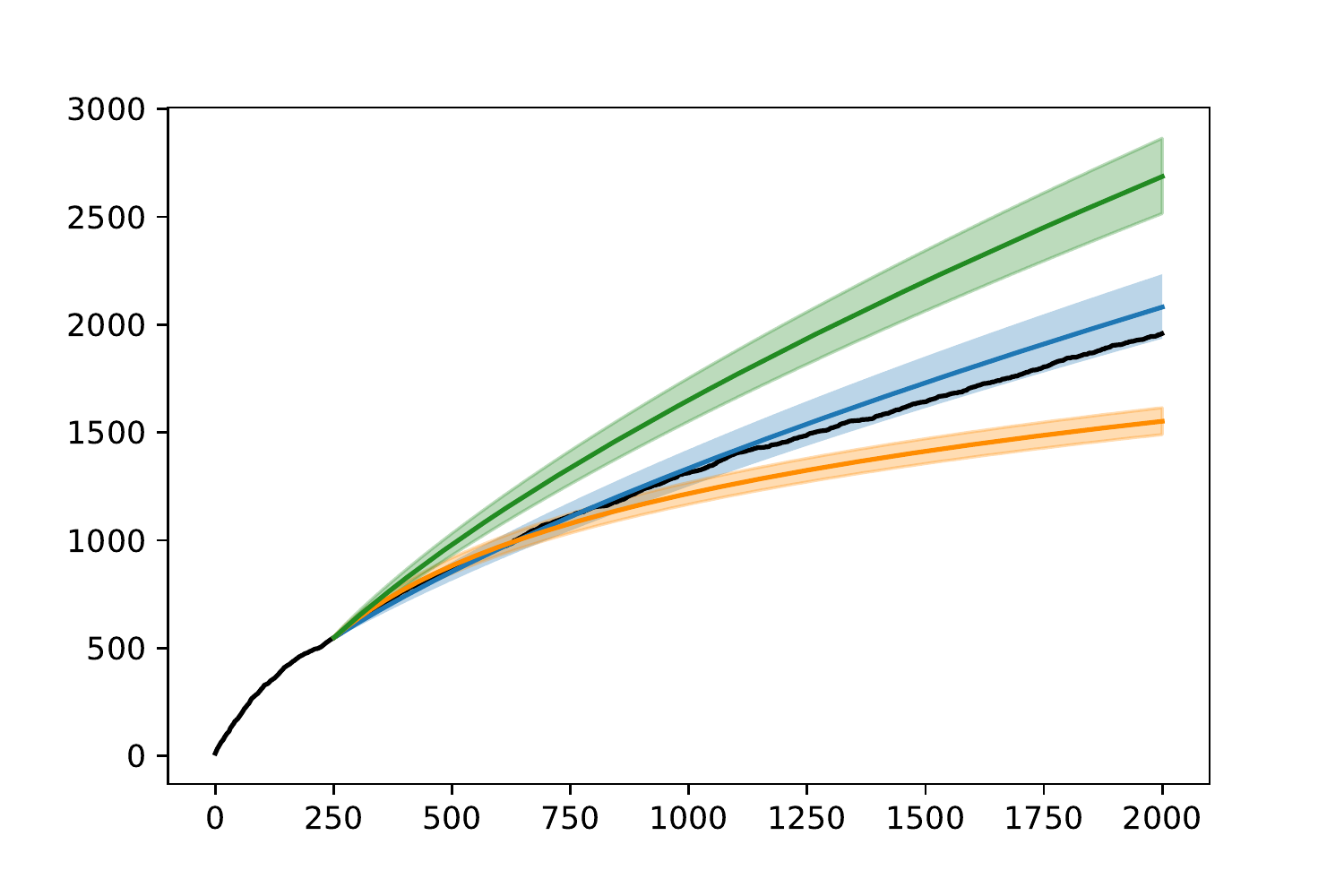}
	\caption{Simulation II}
     \end{subfigure}
    \caption{
        Number of displayed traits (black line) and predicted number of traits under NB-ST-SP (blue line), NB-Ga (orange line), and SB-SP (green line), with 95\% pointwise credible bands (shaded regions).}
    \label{fig:simulation}
\end{figure}

We consider two different data-generating processes. In simulation (I), we simulate data from the NB-ST-SP having fixed $r=10$, $c=60$, and $\alpha =0.3$, using the generative scheme reported in \Cref{sec:nb_restaurant}.
In simulation (II) we consider a sequence $(q_k)_{k \geq 1}$ from the Zipf distribution, i.e. $q_k = (1+k)^{-\xi}$, $\xi=1.5$, and simulate $A_{i, k} \mid q_k$ i.i.d. from a negative binomial distribution with parameters $r=10$ and $1 - q_k$.

We fit the datasets with the NB-ST-SP, NB-Ga, and SB-SP models. 
For the latter model, we preprocess the simulated data to be a binary matrix keeping only track of the presence-absence of traits in the observations.
We keep the parameter $r=10$ fixed and do not estimate it. In this case, the integral $I(r,k)$ in \eqref{eq:i_rk} can be expressed as a sum of Beta functions as shown in \Cref{sec:app_irk}. Estimation of $r$ requires a numerical approximation to $I(r,k)$ which we discuss in Section \ref{sec:text}.

For all simulations, we generate $N=2000$ observations and use the first $n=250$ to estimate the hyperparameters in the models.
Then, we focus on the prediction of $U^{m}_n$, that is the number of new traits displayed in an additional sample of size $m=1, \ldots, 1750$. See  \Cref{sec:app_marg_nb_gamma} for the marginal distribution of $Z_{1:n}$ and the distribution of $U_{n}^{m}$ under the NB-Ga process.
Figure \ref{fig:simulation} summarizes the posterior findings. 
In both settings, NB-ST-SP correctly estimates the distribution of $U_{n}^{m}$. In Simulation I, both NB-Ga and SB-SP underestimate the number of new traits, with SB-SP having greater predictive performance than NB-Ga.
In Simulation II, NB-Ga underestimates the number of new traits, while SB-SP overestimates it, with NB-Ga having a greater predictive performance than SB-SP.

\subsection{Nonparametric ``naive Bayes'' text classification}\label{sec:text}

Consider $g$ corpora of documents $Z_{j, i}, \ldots, Z_{j, n_j}$. For instance, $g=3$ and the $Z_{j, i}$'s are movies reviews subdivided into ``good'', ``neutral'', and ``bad'', the $Z_{j,i}$'s are newspaper articles and which are subdivided by topic (e.g., foreign politics, technology, news stories, and so on).
As common for text classification, we make the ``bag of words'' assumption, i.e., the order in which words appear in a document is irrelevant for the purpose of topic detection. 
Therefore, we can represent each document as a collection of words (i.e., traits) and counts of how many times each word appears in the document (i.e., levels of association). Introducing a corpus-specific random measure $\mu_j$, we have that a suitable model under the bag of words assumption is 
\[
    Z_{j, i} \mid \mu_j \quad \simiid \quad \mathrm{TrP}(G_A, \mu_j), \qquad i=1, \ldots, n_j.
\]
The model is then completed by choosing a suitable distribution $G_A$ and a prior for $(\mu_1, \ldots, \mu_g)$ discussed below.
To classify a new document $Z^*$ into one of the topics $1, \ldots, g$ we follow \cite{Zho(16)} and consider the following ``naive Bayes'' rule:
\[
    \mathrm{Pr}(Z^* \text{ belongs to $j$-th class}) \propto \mathrm{Pr}(Z^* \mid Z_{j, i}, \ldots, Z_{j, n_j}), \quad j=1, \ldots, g
\] 
where the term on the right-hand side is the predictive distribution of $Z^*$ given $Z_{j, i}, \ldots, Z_{j, n_j}$.
Normalizing this vector of size $g$ gives the class assignment probabilities, while taking the argmax gives a point estimate for the classification of document $Z^*$.

We assume independence across different corpora, and model them separately via either the NB-ST-SP or NB-Ga processes. 
As in Section \ref{sec:simulation} we we adopt an empirical Bayesian approach and estimate the hyperparameters by maximizing the (log) marginal likelihood of each corpus.
For NB-ST-SP, this requires the computation of the integrals $I(r, k)$ that we approximate numerically using Gauss-Laguerre quadrature \citep[see, e.g.,][]{Hil(87)} with weight function $e^{-s}$.
Moreover, we can easily compute the predictive probabilities $\mathrm{Pr}(Z^* \mid Z_{j, i}, \ldots, Z_{j, n_j})$ by considering the ratio of the marginal distribution of $(Z_{j, i}, \ldots, Z_{j, n_j}, Z^*)$ and the marginal distribution of $(Z_{j, i}, \ldots, Z_{j, n_j})$, as in \eqref{eq:marg_nb}.
Note that a more principled approach would be to assume a joint distribution for $(\mu_1, \ldots, \mu_g)$, for instance by considering the \mbox{hTSP} process discussed in Section \ref{sec3}.
However, fitting this hierarchical model would require the use of Markov chain Monte Carlo algorithms with non-trivial updates.

We apply our Naive Bayes classifier to the 20-Newsgroup dataset, from which we select seven categories (i.e., classes) of documents for a total of $6,470$ texts of which $3,884$ used for training. 
We exclude stopwords as defined in the default list in the Python package \texttt{scikit-learn} as well as words appearing in less than 3 documents of the training set.
The final vocabulary of the training set consists of $14,369$ unique words.
In particular, we note that, in the training set, the number of distinct words in each class ranges between $6,000$ to $8,800$.
Moreover, the test set contains numerous words not present in the training set.
Specifically, the number of new words in in the $j$-th category (i.e., the number of words that are not present in the $j$-th class for the training set, but are present in the test set) ranges between $850$ and $3,000$.
Hence, once a new document $Z^*$ needs to be classified, it is likely that it will contain several words that have not appeared in $Z_{j,1}, \ldots, Z_{j, n_j}$ for each $j = 1, \ldots, g$.
Therefore, we expect that the increased flexibility of our model yields a superior predictive performance on the test set.

Figures \ref{fig:prediction} and \ref{fig:pred_gamma} shows the classification of documents of the training and test set unader NB-ST-SP and NB-Ga, respectively. The $(i, j)$-th entry of each matrix corresponds to the probability that the $i$-th document belongs to the $j$-th class.
We compute a point prediction by taking the argmax over these probabilities. Using this rule, our classifier achieves $99\%$ accuracy on the training set and $88\%$ accuracy on the test set, while the one using the gamma process prior achieves almost a perfect accuracy on the training set but shows a poorer performance on the test set, where accuracy is around $83\%$.

\begin{figure}[ht]
    \centering
    \includegraphics[width=\linewidth]{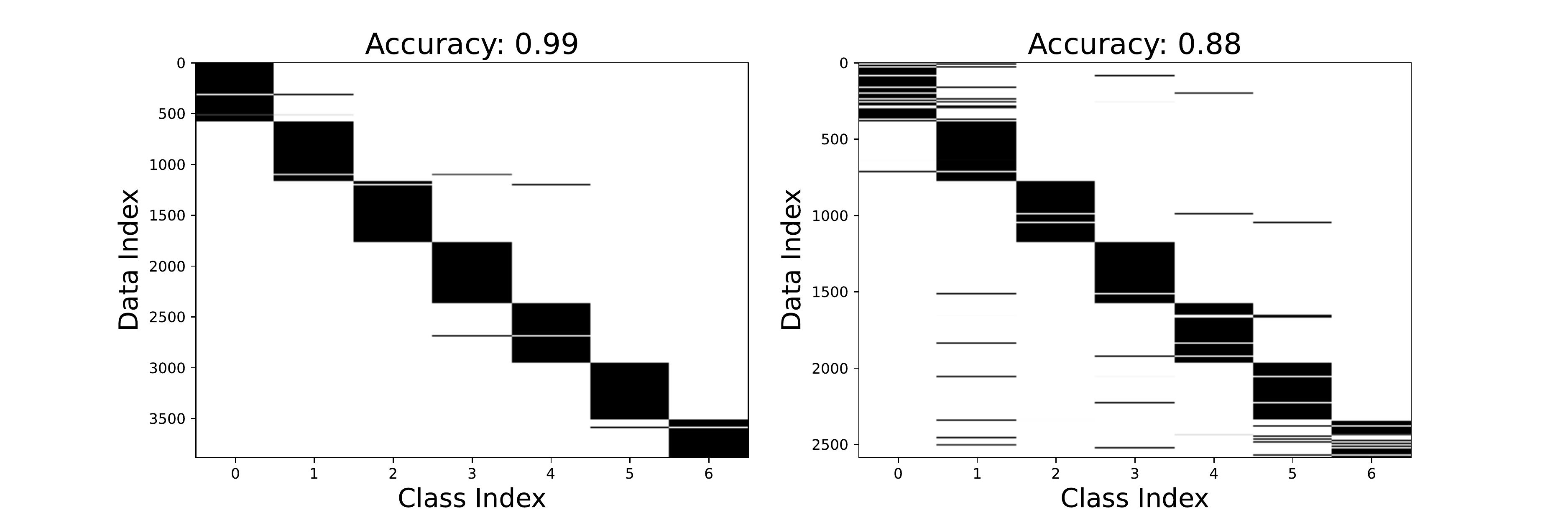}
    \caption{Class assignment probabilities for each document (one for each row) in the training (left) and test (right) set, for our model.  Documents are sorted by their corpus, so that perfect classification corresponds to block-diagonal like matrices.}
    \label{fig:prediction}
\end{figure}

\begin{figure}[ht]
    \centering
    \includegraphics[width=\linewidth]{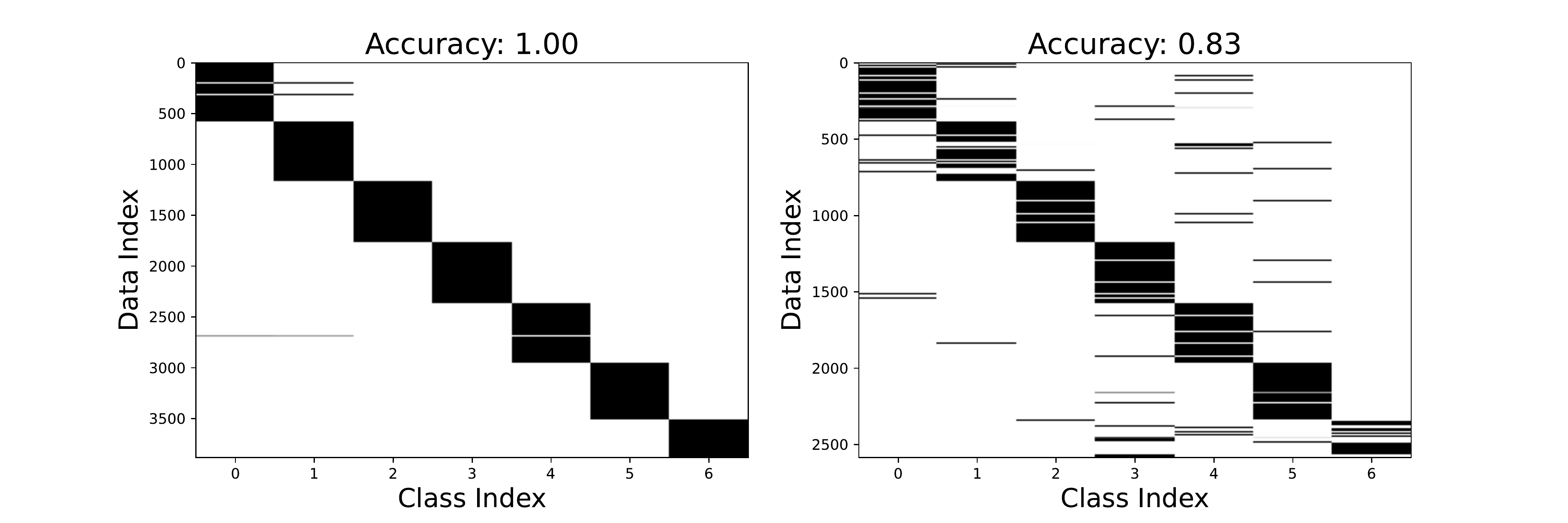}
    \caption{Class assignment probabilities for each document (one for each row) in the training (left) and test (right) set, for the gamma prior model. Documents are sorted by their corpus, so that perfect classification corresponds to block-diagonal like matrices.}
    \label{fig:pred_gamma}
\end{figure}

\section{Discussion}\label{sec5}

We introduced and investigated the class of T-SP priors for BNP analysis of trait allocations, showing that they are characterized by a richer predictive structure than CRM priors, while maintaining the same analytic tractability and ease of interpretation as CRM priors. Then, as an extension of T-SP priors to trait allocations with multiple groups of data or subpopulations, we presented the class of hierarchical T-SP priors, and developed their posterior analysis. The effectiveness of T-SP priors is showcased through an empirical analysis on synthetic and real data. In particular, in an extrapolation task for the number of newly displayed traits in an additional sample, we found that trait models based on T-SP priors outperform models based on CRMs. Moreover, we also show how trait models provide better predictive performance than feature allocation models, which disregard the levels of association and consider only the presence or absence of the traits. In a text classification task on a real dataset with more than six thousand documents and 14 thousand unique words, we show that modeling the word-document counts using TSPs results in a better predictive performance than assuming a CRM prior, hence demonstrating the practical need of models with more flexible predictive structure in real world data analysis problems.

SP priors was introduced in \citet{Jam(15)} as a generalization of the Beta process prior for BNP analysis of feature allocations, and, most recently, they have been applied by \citet{Cam(23)} in BNP inference for the unseen-feature problem, first showing their great potential in enriching the predictive structure of CRM priors. By introducing the class of T-SP priors, our work somehow completes the work of \citet{Jam(15)}, making available a comprehensive framework for BNP analysis of feature and trait allocations, which relies on a broad class of priors that enrich the predictive distribution of CRM priors, while maintaining analytical tractability and ease of interpretation. CRM priors, and in particular the Beta and the Gamma process priors, have been widely used in BNP analysis of feature and trait allocations, with a broad range of applications, e.g., network analysis, analysis of differential gene expression and high-throughput sequencing data, and topic modeling to cite a few. In all these contexts, we believe that T-SP priors may be more effective than CRM priors, as they allow to better exploit the sampling information in posterior inferences.

\FloatBarrier

\appendix

\section{Proofs of the main results}\label{proofs}

We start by a fundamental lemma that will serve as the base in the proofs of all our general results.
\begin{lemma}\label{lemma:conditional_levy}
Let $\mu \sim \tsp(\theta \rho B_0, h, T)$. Then, conditionally to $\Delta_{1, h}$, $\mu$ is a completely random measure with L\'evy intensity
\begin{equation}\label{eq:conditional_levy}
    \frac{\rho(\Delta_{1, h} T^{-1}(s))}{T^\prime(T^{-1})(s)} \Delta_{1, h} \indicator_{(T_0, T_1)}(s) \dd s \theta B_0(\dd x)
\end{equation}
\end{lemma}
\begin{proof}
Consider the random measure $\mu_{\Delta_1}$ defined as
\[
    \mu_{\Delta_1} = \sum_{k \geq 1} T\left(\frac{\Delta_{k+1}}{\Delta_1}\right) \delta_{w_{k+1}}
\]
where $\Delta_1 > \Delta_2 > \cdots$ are the ordered jump of a completely random measure and the $w_{k}$'s its support points. Recall that a $\tsp$ is obtained by a change of measure of $\Delta_1$.
Therefore, it suffices to show that $\mu_{\Delta_1}$ given $\Delta_1$ is a completely random measure with intensity
\[
    \frac{\rho(\Delta_{1} T^{-1}(s))}{T^\prime(T^{-1})(s)} \Delta_{1} \indicator_{(T_0, T_1)}(s) \dd s \theta B_0(\dd x).
\]

As shown in \cite{Jam(15)}, $\Delta_2, \Delta_3, \ldots$, conditionally to $\Delta_1$ are the points of a Poisson process with intensity $\rho(s) \indicator_{(0, \Delta_1)}(s) $.
Consider now the Laplace transform of $\mu_{\Delta 1}$, for measurable $f: \W \rightarrow \R_+$
\begin{align*}
    \E\left[e^{- \int_\W f(x) \mu_{\Delta_1}(\dd x)} \mid \Delta_1 \right] &= \E \left[\exp \left\{ - \sum_{k \geq 1} f(w_k) T\left(\frac{\Delta_{k+1}}{\Delta_1} \right) \right\} \mid \Delta_1 \right] \\
    &= \exp\left\{ - \int_\W \int_{\R_+} (1 - e^{-f(w) T(s \Delta^{-1})}) \rho(s) \indicator_{(0, \Delta_1)}(s) \dd s \theta B_0(\dd x) \right\}
\end{align*}
where the second equality follows from the representation of the Laplace transform for a marked Poisson point process with points $(\Delta_{k+1}, w_{k+1})_{k \geq 1}$. The change of variables $y=s \Delta^{-1}$ and $s = T(y)$ yields
\begin{align*}
    \E\left[e^{- \int_\W f(x) \mu_{\Delta_1}(\dd x)} \mid \Delta_1 \right] &= \exp\left\{ - \int_\W \int_0^1 (1 - e^{-f(w) T(y)}) \rho(y \Delta_1) \Delta_1 \dd y \theta B_0(\dd x) \right\} \\
    &= \exp\left\{ - \int_\W \int_{T_0}^{T_1} (1 - e^{-f(w) s}) \rho(T^{-1}(s) \Delta_1) \Delta_1 \frac{1}{T^\prime(T^{-1}(s))}\dd s \theta B_0(\dd x) \right\}
\end{align*}
which concludes the proof. 
\end{proof}

\subsection{Proof of \eqref{eq:finite_traits}}\label{app:proof_finite_traits}

To prove that this is a sufficient condition, we observe that the non-negativeness of the $B_k$'s entails that it is sufficient to ask that the expected value of their sum is finite. Hence,
\begin{align*}
    \E\left[ \sum_{k \geq 1} B_{k} \right] &= \E_{\Delta_{1, h}}\left[ \E \left[ \sum_{k \geq 1} \pi_A(\tilde \rho_k) \mid \Delta_{1, h} \right] \right] \\
    &= \E_{\Delta_{1, h}}\left[ \int_{\W} \int_{T_0}^{T_1} \pi_A(s) \frac{\rho(\Delta_{1, h} T^{-1}(s))}{T^\prime(T^{-1})(s)} \Delta_{1, h} \dd s \theta B_0(\dd x) \right]
\end{align*}
where the second equality follows from Campbell's theorem and an application of Lemma \ref{lemma:conditional_levy}. Then the proof follows by integrating with respect to $B_0(\dd x)$.

\subsection{Proof of Theorem \ref{marg_tsp}}\label{app:proof_marg}

Conditionally to $\Delta_{1, h}$, the random measure $\mu$ is a CRM. Hence, by Proposition 3.1 in \cite{Jam(17)} the distribution of $(Z_1, \ldots, Z_n) \mid \Delta_{1,h}$ equals 
\[
    \theta^{k_{n}}\exp \left\{ -\theta \sum_{i=1}^n \varphi_{i}(\Delta_{1, h}) \right\} \prod_{l=1}^{k_{n}} \int_{T_0}^{T_1} \rho_{n - m_l}(s, \Delta_{1, h}) \pi_A(s)^{m_l} \prod_{i \in \mathcal{B}_l} \left[\tilde G_{A^\prime}(\dd a_{i, k}; s) \right] \dd s,
\]
where, by virtue of \Cref{lemma:conditional_levy}
\[
    \varphi_k(\zeta) =  \int_{T_0}^{T_1} \pi_A(s) \rho_{k-1}(s, \zeta) \dd s, \quad  \rho_k(s, \zeta) = \left(1 - \pi_A(s) \right)^{k} \frac{\zeta \rho(\zeta T^{-1}(s))}{T^\prime(T^{-1}(s))} I_{(T_0, T_1)}(s).
\]
Then consider the summation of the $\varphi_i$'s. An application of Fubini's theorem yields
\begin{align*}
    \sum_{i=1}^n \varphi_{i}(\Delta_{1, h}) &=  \int_{T_0}^{T_1} \pi_A(s) \sum_{i=1}^n (1 - \pi_A(s))^{i-1} \frac{\zeta \rho(\zeta T^{-1}(s))}{T^\prime(T^{-1}(s))}\\
    &= \int_{T_0}^{T_1} \pi_A(s) \frac{1 - (1 - \pi_A(s))^n}{1 - (1 - \pi_A(s))} \frac{\zeta \rho(\zeta T^{-1}(s))}{T^\prime(T^{-1}(s))} = \psi_n(\zeta).
\end{align*}

\subsection{Proof of Theorem \ref{teo:post_tsp}}\label{app:proof_post}

The posterior of $\Delta_{1, h}$ follows by noticing that the marginal distribution of $Z_{1:n}$ is 
\[
    \int_0^\infty e^{- \theta \psi_n(\zeta)} \prod_{l=1}^k \int_{T_0}^{T_1} \rho_{n - m_l}(s, \zeta) \pi_A(s)^{m_l} \prod_{i \in \mathcal{B}_l} \left[\tilde G_{A^\prime}
(\dd a_{i, k} \mid s) \right] \dd s  \theta^k f_{\Delta_{1, h}}(\zeta) \dd \zeta.
\]
An application of Bayes' theorem yields \eqref{eq:post_delta}.

The posterior distribution of $\mu$ follows by first conditioning to $\Delta_{1, h}$, and applying \Cref{teo:post_crm} to get the distribution of $\mu \mid Z_1, \ldots, Z_n, \Delta_{1,h}$, and then marginalizing with respect to \eqref{eq:post_delta}. The same strategy can be used to derive the predictive distribution of $Z_{n+1}$.

\section{Details about the examples}\label{app:example_details}

\subsection{Calculations for the Poisson case}\label{app:ex_pois}

To obtain the distribution of $\Delta_{1, h}^{-\alpha}$ given $Z_1, \ldots, Z_n$, from the definition of $F(n, q, r, \alpha)$ and \eqref{eq:post_delta} we get
\begin{multline*}
    f_{\Delta_{1, h}^{-\alpha} \mid Z}(\zeta) \propto\frac{(\alpha \theta \zeta^{-\alpha})^k}{\prod_{l=1}^k \prod_{i \in B_l} a_{i, l}!} \exp\left\{- \theta \alpha \zeta^{-\alpha} I(r, n)\right\} \\ \times \prod_{l=1}^k F(n, q_l, r \alpha) \theta^{c}{\Gamma(c)} (\zeta^{-\alpha})^{c-1} e^{- \theta \zeta^{- \alpha}} \dd \zeta^{-\alpha}
\end{multline*}
where we recognize the kernel of a Gamma distribution with parameters $c+k_n$ and $\theta(1 + \alpha I(r, n))$.

As far as the marginal distribution of $Z_{1:n}$ is concerned, from Theorem \ref{marg_tsp}, conditionally to $\Delta_{1, h} = \zeta$, the marginal distribution of $Z_{1:n}$ equals
\begin{multline*}
        \theta^k \exp\left\{-  \theta \alpha \zeta^{-\alpha} I(r, n)\right\} \\ 
        \times \prod_{l=1}^k \int_{\R_+} e^{-rs(n - m_l)} (1 - e^{-rs})^{m_l} e^{-s} \alpha \zeta^{-\alpha}(1 - e^{-s})^{- 1- \alpha} \prod_{i \in B_l} \frac{e^{-rs} rs^{a_{il}}}{a_{il}! (1-e^{-rs})}
\end{multline*}
Letting $q_l := \sum_{i \in B_l} a_{i,l}$ and recalling the definition of $F(n, q_l, r, \alpha)$ we get that the expression above reduces to
\[
    \frac{(\alpha \theta \zeta^{-\alpha})^k}{\prod_{l=1}^k \prod_{i \in B_l} a_{i, l}!} \exp\left\{-  \theta \alpha \zeta^{-\alpha} I(r, n)\right\} \prod_{l=1}^k F(n, q_l, r \alpha).
\]
Integrating with respect to $\Delta_{1, h}^{-\alpha} \sim \mbox{Gamma}(c, \theta)$ leads to the marginal of $Z_1, \ldots, Z_n$ being equal to
\begin{align*}
   & \int_{\R_+} \frac{(\alpha \theta \zeta^{-\alpha})^k}{\prod_{l=1}^k \prod_{i \in B_l} a_{i, l}!} \exp\left\{-  \theta \alpha \zeta^{-\alpha} I(r, n)\right\} \prod_{l=1}^k F(n, q_l, r \alpha) \frac{\theta^{c}}{\Gamma(c)} (\zeta^{-\alpha})^{c-1} e^{- \theta \zeta^{- \alpha}} \dd \zeta^{-\alpha} \\
   & \qquad \qquad = \frac{(\alpha \theta)^k \theta^c \prod_{l=1}^k F(n, q_l, r \alpha)}{\Gamma(c) \prod_{l=1}^k \prod_{i \in B_l} a_{i, l}!} \int (\zeta^{-\alpha})^{k + c - 1} \exp\left\{- \theta \alpha \zeta^{-\alpha}  I(r, n) + \theta \right\} \dd \zeta^{-\alpha}
\end{align*}
and the result follows by recognizing the kernel of a $\mbox{Gamma}(k + c, \theta(1 + \alpha \sum I(r, i)))$ in the integral.

\subsection{Calculations for the Negative Binomial case}

We proceed along the same lines of \Cref{app:ex_pois}, substituting $\tilde G_{A^\prime}=\binom{a + r - 1}{a}(1 - e^{-s})^a e^{-sr} (1 - e^{-sr})^{-1}$.
Then conditionally to $\Delta_{1, h} = \zeta$ one obtains that the marginal distribution of the sample is
    \begin{multline*}
        (\alpha \theta \zeta^{-\alpha})^k \prod_{l=1}^k \prod_{i \in B_l} \binom{a_{i, l} + r - 1}{a_{i, l}} \exp\left\{- \theta \alpha \zeta^{-\alpha} I(r, n)\right\} \times \\
        \times \prod_{l=1}^k \int_{\R_+} (1 - e^{-s})^{-1-\alpha + q_l} e^{- s r n} e^{-s} \dd s
    \end{multline*}
    where the integrals can be evaluated using the change of variable $y= e^{-s}$, leading to
    \[
        \int_{\R_+} (1 - e^{-s})^{-1-\alpha + q_l} e^{- s r n} e^{-s} \dd s = \mbox{Beta}(rn - 1, -\alpha + q_l).
    \]
    Finally, integrating with respect to $\Delta_{1, h}^{- \alpha} \sim \mbox{Gamma}(c, \theta)$ yields the marginal distribution.
    
The conditional distribution of $\Delta_{1, h}^{-\alpha}$ follows exactly as in \Cref{app:ex_pois}.

Consider now the density of $J^*_l$ \eqref{eq:post_jumps}. Standard computations lead to
\begin{align*}
    f_{J^*_l}(s) = \frac{1}{B(rn - 1, q_l - \alpha)} (1 - e^{-s})^{q_l - \alpha - 1} e^{-srn} e^{-s},
\end{align*} 
we note in particular that this density does not depend on $\Delta_{1,h}$.
This leads to the predictive distribution for $A_{n+1, l}$, $l=1, \ldots, k_n$, that is supported on $\{0, 1, 2, \ldots\}$ such that
\begin{equation}\label{eq:nb_new_trait}
    \mathrm{Pr}(A_{n+1, l} = k \mid Z_1, \ldots, Z_n) = \binom{k + r - 1}{k} \frac{B\left(r(n+1) - 1, q_l + k - \alpha\right)}{B(rn - 1, q_l - \alpha)}
\end{equation} 
which does not belong to a known parametric family.

\subsection{Gaussian Spike-and-Slab case}

The examples in Section \ref{sec2} were based on $G_A$ being a p.m.f. over $\{0, 1, \ldots\}$. 
Now, we consider an alternative construction where we specify directly $\pi_A$ and $\tilde G_{A^\prime}$ directly as a Gaussian spike and slab as in Section 4.4 of \cite{Jam(17)}. That is
\[
    \pi_A(s) = 1 - e^{-s}, \qquad \tilde G_{A^\prime}(y; s, \eta_k) = \frac{\sqrt s}{\sqrt{2 \pi}} e^{-\frac{s}{2} (y - \eta_k)^2 }
\]
where the atom-specific parameters $\eta_k$ are assumed fixed.
A priori, we assume $\mu \sim \tsp \left(\alpha s^{-1-\alpha} B_0(\dd x), h_{c}, -\log(1 - s) \right)$ as in the Poisson case.

Since $\pi_A(s)$ is identical to the expression in the Poisson case with $r=1$, we have that $\rho_k$ and $\phi_k$ have the same expression here.
Hence, conditionally to $\Delta_{1,h} = \zeta$, the marginal distribution of $Z_1, \ldots, Z_n$ equals
\begin{align*}
    & e^{- \theta \alpha \zeta^{-\alpha}  I(1, n)} \prod_{l=1}^k \int_{\R_+} \theta \alpha \zeta^{-\alpha} e^{-s (n-m_l)} (1 - e^{-s})^{m_l} e^{-s} (1 - e^{-s})^{-1-\alpha} \prod_{i \in \mathcal{B}_l} \frac{\sqrt s}{\sqrt{2 \pi}} e^{-\frac{s}{2} (y_{i, l} - \eta_l)^2 }  \dd s\\
    & \qquad =  e^{- \theta \alpha \zeta^{-\alpha}  I(1, n)} \theta^k \alpha^k (\zeta^{-\alpha})^k \prod_{l=1}^k (2 \pi)^{m_l / 2} \\
    & \hspace{4cm} \int_{\R_+} e^{- s (n - m_l)} e^{-s} (1 - e^{-s})^{-1-\alpha + m_l} s^{m_l / 2} e^{- \frac{s}{2} \sum_{i \in \mathcal{B}_l} (y_{i, l} - \eta_l)^2} \dd s
\end{align*}
where the above integral does not possess an analytic expression. Let
\[
    G(n, \mathcal A_{l}, \eta_l, \alpha) = (2 \pi)^{m_l / 2} \int_{\R_+} e^{- s (n - m_l)} e^{-s} (1 - e^{-s})^{-1-\alpha + m_l} s^{m_l / 2} e^{- \frac{s}{2} \sum_{y \in A_{l}} (y - \eta_l)^2} \dd s
\]
(note that $|\mathcal A_{l}| = m_l$).
Then we have that the marginal distribution of $Z_1, \ldots, Z_n$ conditionally to $\Delta_{1,h} = \zeta$ equals
\[
    e^{- \theta \alpha \zeta^{-\alpha} I(1, i)} \theta^k \alpha^k (\zeta^{-\alpha})^k \prod_{l=1}^k G(n, \mathcal A_{l}, \eta_l, \alpha).
\]
Integrating with respect to  $\Delta_{1, h}^{-\alpha} \sim \mbox{Gamma}(c, \theta)$ leads to the following expression for the marginal distribution
\[
    \frac{\Gamma(k + c)}{\Gamma(c)} \frac{\alpha^k \prod_{l=1}^k G(n, \mathcal A_{l}, \eta_l, \alpha)}{\left( 1 + \alpha  I(1, i)\right)^{c + k}}. 
\]

\section{Further details on the simulations}\label{sec:app_simu}

\subsection{Indian Buffet Process for the negative-Binomial trait process}\label{sec:nb_restaurant}

We describe here a generative process to simulate from the negative-Binomial process discussed in Section \ref{sec:nb}.
This is based on the predictive distribution in \eqref{pred_tsp} and the discussion in Section \ref{sec:nb}.
In particular, we recall that a new customer displays a number of new traits that is negative-binomial distributed, and assigns to the previously observed trait values from \eqref{eq:nb_new_trait}.
What is left to discuss is the distribution of the association level to the new traits that the new customer displays.

Consider first customer number 1. From Propostion 3.3 in \cite{Jam(17)}, conditionally to $\Delta_{1, h}$, the $A_{1, k}$'s such that $A_{1, k} > 0$ (note that there are $K_1$ of these, such that $K_1 \mid \Delta_{1, h}$ is Poisson distributed with parameter $\psi_1(\Delta_{1,h})$) are independent and identically distributed such that 
\begin{align*}
    A_{1, k} \mid \Delta_{1, h}, H_{1, k} = s &\quad \sim \quad \tilde G_{A^\prime}(s) \\
    H_{1, k} \mid \Delta_{1, h} &\quad \sim \quad f_{H | \Delta_{1, h}}(s) \propto \pi_A(s) \rho_0(s, \Delta_{1, h})
\end{align*}
where 
\[
    \rho_0(s, \zeta) = \frac{\zeta \rho(\zeta T^{-1}(s))}{T^\prime(T^{-1}(s))} I_{(T_0, T_1)}(s).
\]
Marginalizing with respect to $\Delta_{1, h}$ we have that $K_1$ is negative binomial distributed with parameters $c$ and $(1 + \alpha I(r, 0)) / (1 + \alpha I(r, 1))$. Moreover, the $H_{1, k}$ are independent with marginal density 
\[
    f_{H}(s) \propto (1 - e^{-rs}) (1 - e^{-s})^{-1-\alpha} e^{-s}
\]

After $n$ customer have entered, having chosen $k_n$ dishes with scores $\mathcal A_{l}$, $l=1, \ldots, k_n$ and associated index sets $\mathcal B_l$, customer $n+1$ assigns to each of the dishes previously served a score $A_{n+1, l} \mid J^*_l \sim G_{A}(\cdot \mid J^*_l)$. 
Marginalizing with respect to $J^*_l$, the marginal distribution of $A_{n+1, l}$ is \eqref{eq:nb_new_trait}, $l=1, \ldots, k_n$.
Then, she chooses $K_{n+1}$ new dishes such that $K_{n+1}$ is negative binomial distributed with parameters $c+k_n$ and $(1 + \alpha I(r, n)) / (\alpha \tilde I(r, n+1))$, and assigns to these a score $A^{\prime}_{n+1, h}$ $h = 1, \ldots, K_{n+1}$ such that 
\begin{align*}
    A^\prime_{n+1, k} \mid \Delta_{1, h}, H_{1, k} = s &\quad \sim \quad \tilde G_{A^\prime}(s) \\
    H_{n+1, k} \mid \Delta_{1, h} &\quad \sim \quad f_{H | \Delta_{1, h}}(s) \propto \pi_A(s) \rho_n(s, \Delta_{1, h})
\end{align*}
Marginalizing with respect to $\Delta_{1, h}$, we get that the $H_{n+1, k}$'s are independent with marginal density
\[
    f_{H_{n+1}}(s) \propto (1 - e^{-rs}) (1 - e^{-s})^{-1-\alpha} e^{-s(rn + 1)}.
\]
The density of the $H_{n+1, k}$'s does not belong to a known parametric family, but can be sampled using rejection sampling, using an exponential distribution with parameter 1 as envelope. Similarly, sampling $A^{\prime}_{n+1, k} \sim \tilde G_{A^\prime}(s)$ is performed via rejection sampling using the negative binomial with parameter $s$ as envelope.
Instead, to sample from \eqref{eq:nb_new_trait} we found it more convenient to evaluate this density on a grid $\{0, 1, \ldots, M\}$ and sample from a discrete distribution on the grid. This introduces a small error which is negligible as long as $M$ is large.

\subsection{The negative-Binomial IBP with Gamma process prior}\label{sec:app_marg_nb_gamma}

From \cite{Jam(17)}, we have that the marginal distribution of the sample is
\[
    e^{- \theta \sum_{i=1}^m \varphi_i} \theta^k \prod_{l=1}^k \int_{0}^{\infty} (1 - \pi_A(s))^{n - m_l} \pi_A(s)^{m_l} \prod_{i \in \mathcal{B}_l} \tilde G_{A^\prime}(\dd a_{i,k} \mid s) \rho(s) \dd s B_0(\dd w^*_k)
\]
In the case of $G_A$ as in \eqref{eq:nb_pmf}, we have $\pi_A(s) = (1 - e^{-sr})$ and
\[
    \varphi_i = \int_{\R_+} (1 - e^{-sr}) e^{-s r i} e^{-s} s^{-1} \dd s = \log\left(\frac{1 + r + ir}{1 + ir}\right)
\]
Leading to the following expression for the marginal
\[
    \theta^k \prod_{i=1}^m \left(\frac{1 + r + ir}{1 + ir}\right)^{-\theta} \prod_{l=1}^k \left[ \prod_{i \in \mathcal{B}_l} \binom{a_{i, l} + r - 1}{a_{i, l}}\right] \int_{\R_+} e^{-srn} e^{-s} (1 - e^{-s})^{q_l} s^{-1} \dd s B_0(\dd \omega^*_k)
\]

Moreover, the number of new traits $U_{n}^{m}$ displayed in an additional sample of size $m$ is Poisson distributed with parameter
\[
	\theta \sum_{j=1}^m \varphi_{n+j}  = \theta \log\left( \frac{1 + (n+m) r}{1 + nr} \right)
\]

\subsection{A close form expression for $I(r, k)$ when $r$ is integer}\label{sec:app_irk}

When $r$ is a positive integer, $I(r, k)$ can be expressed as 
\begin{align*}
    I(r, k) &= \int_{\R_+} (1-  e^{-rks}) (1 - e^{-s})^{-1-\alpha} e^{-s} \dd s \\
    &= \int_0^1 (1 - (1 - y)^{rk}) y^{-1-\alpha} \dd y \\
    &= \int_0^1 \sum_{i=0}^{rk - 1} (1 - y)^i y y^{-1-\alpha} \dd y \\
    &= \sum_{i=1}^{rk} \int_{0}^{1} (1 - y)^{i - 1} y^{-\alpha + 1 - 1} \dd y = \sum_{i=1}^{rk} B(1 - \alpha, i).
\end{align*}
The second equality follows from the change of variable $y = 1 - e^{-s}$, and the second by writing $(1 - (1 - y)^{rk})$ as the partial sum of a geometric series.

\FloatBarrier

\section*{Acknowledgement}
Mario Beraha and Stefano Favaro received funding from the European Research Council (ERC) under the European Union's Horizon 2020 research and innovation programme under grant agreement No 817257. Stefano Favaro gratefully acknowledge the financial support from the Italian Ministry of Education, University and Research (MIUR), ``Dipartimenti di Eccellenza" grant 2018-2022.

\end{document}